\documentclass[11pt]{article}

\usepackage{etex}
\usepackage[thinlines]{easytable}
\usepackage[usenames,dvipsnames]{xcolor}
\usepackage{algorithmic, algorithm}
\usepackage{amsmath,amsfonts,amssymb,tikz}
\usepackage[margin=1in]{geometry}
\usepackage{amsthm} % includes the proof environment
\usepackage{csquotes}
\usepackage[all]{xy}
\usepackage{prettyref} % Use this package to refer to theorems, conjectures and etc

\newtheorem{theorem}{Theorem}[section]
\newtheorem{observation}[theorem]{Observation}

\newtheorem{lemma}[theorem]{Lemma}
\newtheorem{corollary}[theorem]{Corollary}

\newtheorem{claim}[theorem]{Claim}
\newtheorem{problem}[theorem]{Problem}

\newtheorem{definition}[theorem]{Definition}
\newtheorem{remark}[theorem]{Remark}

\newtheorem{fact}[theorem]{Fact}

\makeatletter
\newtheorem*{rep@theorem}{\rep@title}
\newcommand{\newreptheorem}[2]{%
\newenvironment{rep#1}[1]{%
 \def\rep@title{#2 \ref{##1}}%
 \begin{rep@theorem}}%
 {\end{rep@theorem}}}
\makeatother

\newreptheorem{theorem}{Theorem}
\newreptheorem{lemma}{Lemma}

\usepackage{thmtools}
\usepackage{thm-restate}
\newcommand{\pref}{\prettyref}
\newrefformat{lem}{Lemma \ref{#1}}
\newrefformat{thm}{Theorem \ref{#1}}
\newrefformat{cha}{Chapter \ref{#1}}
\newrefformat{sec}{Section \ref{#1}}
\newrefformat{tab}{Table \ref{#1}}
\newrefformat{fig}{Figure \ref{#1}}
\newrefformat{conj}{Conjecture \ref{#1}}
\newrefformat{fact}{Fact \ref{#1}}
\newrefformat{cla}{Claim \ref{#1}} 

\newcommand{\E}{{\mathbb{E}}}

\newcommand{\e}{\mathrm{e}}
\newcommand{\eps}{\varepsilon}

\newcommand{\cA}{\mathcal{A}}

\newcommand{\cD}{\mathcal{D}}

\newcommand{\cI}{\mathcal{I}}

\newcommand{\cN}{\mathcal{N}}

\DeclareMathOperator{\polylog}{polylog}

\usepackage{xspace}

\newcommand{\OPT}{{\operatorname{OPT}\xspace}}

\newcommand{\Otilde}{\tilde{O}}

\newcommand{\NP}{{\bf NP}}
\newcommand{\BPTIME}{{\bf BPTIME}}
\newcommand{\TSAT}{$\mathbf{3SAT}$\xspace}
\newcommand{\TCSP}{$\mathbf{2CSP}$\xspace}

\newcommand{\Div}[2]{{\mathsf{D_{KL}}\left(#1 \Big\| #2\right)}}

\DeclareMathOperator{\den}{den}

\newcommand{\DKSab}[2]{\mathsf{DkS}\left(#1 , #2\right)}

\newcommand{\EC}[1]{}

\newcommand{\ignore}[1]{}%

\renewcommand{\epsilon}{\varepsilon}

\date{\nonumber}
\begin{document}

\title{ETH Hardness for Densest-$k$-Subgraph with Perfect Completeness}%Approximating Densest $k$-subgraph %is \ETH-hard}
%in $n^{o(\log n)}$-time breaks the Exponential Time Hypothesis}

\author{
Mark Braverman \thanks{Department of Computer Science, Princeton University, email: mbraverm@cs.princeton.edu. Research supported in part by an   an NSF CAREER award (CCF-1149888), a 
Turing Centenary Fellowship, a Packard Fellowship in Science and Engineering, and the Simons Collaboration on Algorithms and Geometry.}
\and
Young Kun Ko \thanks{Department of Computer Science, Princeton University, email: yko@cs.princeton.edu}
\and
Aviad Rubinstein \thanks{Department of Electrical Engineering and Computer Sciences, University of California at Berkeley, email: aviad@eecs.berkeley.edu}
\and
Omri Weinstein  \thanks{Department of Computer Science, Princeton University, email: oweinste@cs.princeton.edu}
}

\thispagestyle{empty}\maketitle\setcounter{page}{0}

\begin{abstract}
% !TeX root = Dks_main_10pg.tex

%We study a special case of the well known Densest $k$-Subgraph problem, 
\ignore{We show that, assuming the (deterministic) Exponential Time Hypothesis, distinguishing between a graph with an induced $k$-clique
and a graph where all induced subgraphs of size $k'\ll k$ have (edge) density at most $1-\eps$, requires $n^{\tilde{\Omega}(\log n)}$ time.
Our hardness result essentially matches the quasi-polynomial algorithms of Feige and Seltser \cite{fs97} and Barman \cite{Barman14}, 
and strengthens previous hardness results for this problem which were based on average-case assumptions. 
Our reduction is inspired by the recent application of the ``birthday repetition" technique \cite{AIM14,BKW15}. Our analysis relies 
on information theoretical machinery and is similar in spirit to analyzing a parallel repetition of two-prover games in which the provers
may choose to answer some challenges multiple times, while completely ignoring other challenges. 
}

We show that, assuming the (deterministic) Exponential Time Hypothesis, distinguishing between a graph with an induced $k$-clique
and a graph in which all $k$-subgraphs have density at most $1-\eps$, requires $n^{\tilde{\Omega}(\log n)}$ time.
%In particular, this is the first result to rule out an additive PTAS for the Densest $k$-Subgraph problem,
Our result essentially matches the quasi-polynomial algorithms of Feige and Seltser \cite{fs97} and Barman \cite{Barman14} for this problem,
and is the first one to rule out an additive PTAS for Densest $k$-Subgraph.
We further strengthen this result by showing that our lower bound continues to hold when, in the soundness case, 
even subgraphs smaller by a near-polynomial factor ($k' = k \cdot 2^{-\tilde \Omega (\log n)}$) are assumed to be at most $(1-\epsilon)$-dense.

Our reduction is inspired by recent applications of the ``birthday repetition" technique \cite{AIM14,BKW15}. Our analysis relies 
on information theoretical machinery and is similar in spirit to analyzing a parallel repetition of two-prover games in which the provers
may choose to answer some challenges multiple times, while completely ignoring other challenges.

\end{abstract}

\thispagestyle{empty}
\newpage

% !TeX root = Dks_main.tex

\section{Introduction} 
\label{sec:intro}

$k$-{\sc Clique} is one of the most fundamental problems in computer
science: given a graph, decide whether it has a fully connected induced subgraph
on $k$ vertices. Since it was proven \NP-complete by Karp \cite{Karp72},
extensive research has investigated the complexity of relaxed versions of this problem.

This work focuses on two natural relaxations of $k$-{\sc Clique} which have received significant 
attention from both algorithmic and complexity communities: 
The first one is to relax ``$k$", i.e. looking for a smaller subgraph:
\begin{problem}[Approximate Max Clique, Informal] \label{prob_1}
Given an $n$-vertex graph $G$, decide whether $G$ contains a clique of size $k$, or all induced cliques of $G$ are of size at most $\delta k$
for some $1> \delta(n)>0$.
\end{problem}
\noindent 
The second natural relaxation is to relax the ``Clique'' requirement,
replacing it with the more modest goal of finding a subgraph that is almost a clique:
\begin{problem}[Densest $k$-Subgraph with perfect completeness, Informal]  \label{prob_2}
Given an $n$-vertex graph $G$ containing a clique of size $k$, find an induced subgraphs of $G$ of size $k$ with (edge) density 
at least $(1-\eps)$, for some $1> \eps > 0$.
%Given a $n$-vertex graph $G$, decide whether $G$ contains a clique of size $k$, or all induced subgraphs of $G$ of size $k$ have (edge) density 
%at most $(1-\eps)$, for some $1> \eps > 0$.
(More modestly, given an $n$-vertex graph $G$, decide whether $G$ contains a clique of size $k$, or all induced $k$-subgraphs of $G$ have 
density at most $(1-\eps)$).
\end{problem}

Today, after a long line of research 
\cite{fglss96, AS98-PCP, ALMSS98-PCP, Hastad96-clique_n^1-eps, Khot01-clique_coloring, Zuckerman07-clique_NP}
we have a solid understanding of the inapproximability
of Problem~\ref{prob_1}. In particular, we know that it is \NP-hard
to distinguish between a graph that has a clique of size $k$, and
a graph whose largest induced clique is of size at most $k'=\delta k$ for $\delta=1/n^{1-\epsilon}$ \cite{Zuckerman07-clique_NP}. 
The computational complexity of the second relaxation (Problem~\ref{prob_2})
remained largely open. There are a couple of quasi-polynomial algorithms that % \mnote{O: why is Problem 1.2 interesting? stronger motivation...?}
guarantee finding a $(1-\eps)$-dense $k$ subgraph in every graph containing a $k$-clique
%(and in particular solve the decision version of Problem~\ref{prob_1})
%distinguish between a graph with a $k$-clique and a graph where every
%$k$-subgraph is $(1-\epsilon)$-sparse 
\cite{fs97,Barman14}\footnote{Barman \cite{Barman14} approximates the {\sc Densest $k$-Bi-Subgraph}
problem. {\sc Densest $k$-Subgraph} can be handled via a simple
modification \cite{Barman15-DkS_vs_DkBS}.%
}, suggesting that this problem is not \NP-hard. Yet we know neither
polynomial-time algorithms, nor general impossibility results for this problem.

In this work we provide a strong evidence that the
aforementioned quasi-polynomial time algorithms for Problem~\ref{prob_2} \cite{fs97,Barman14}
are essentially tight, assuming the 
(deterministic) {\em Exponential Time Hypothesis} (ETH), which postulates that 
any deterministic algorithm for {\sc 3SAT} requires $2^{\Omega(n)}$ time \cite{ETH01}.
In fact, we show that under ETH, both parameters of the above relaxations are simultaneously 
hard to approximate:

\begin{theorem}[Main Result]   \label{thm_main_result}
There exists a universal constant $\epsilon>0$ such that, assuming
the (deterministic) Exponential Time Hypothesis, distinguishing between the following requires time $n^{\tilde{\Omega}\left(\log n\right)}$, where $n$ is the number of vertices of $G$.
\begin{description}
\item [{Completeness}] $G$ has an induced $k$-clique; and
\item [{Soundness}] Every induced subgraph of $G$ size $k'=k \cdot 2^{-\Omega(\frac{\log n}{\log\log n})}$
has density at most $1-\epsilon$,
\end{description}
\end{theorem}

Our result has %important 
implications for two major open problems whose computational complexity remained elusive for more than two decades:
%Problem \ref{prob_2} is a special case of both 
The (general) {\sc Densest $k$-Subgraph} problem, and the {\sc Planted Clique} problem.

The {\sc Densest $k$-Subgraph} problem, $\DKSab{\eta}{\eps}$, is the same as (the decision version of) Problem \ref{prob_2}, except 
that in the ``completeness" case, $G$ has a $k$-subgraph with density $\eta$, and in the ``soundness" case, every $k$-subgraph is of 
density at most $\eps$, where $\eta \gg \eps$. Since Problem \ref{prob_2} is a special case of this problem, our main theorem can also be 
viewed as a new inapproximability result for 
$\DKSab{1}{1-\eps}$. We remark that the aforementioned quasi-polynomial algorithms for the ``perfect completeness" regime completely 
break in the sparse regime, and indeed it is believed  
that $\DKSab{n^{-\alpha}}{n^{-\beta}}$ (for $k=n^{\eps}$) in fact requires much more than quasi-polynomial time \cite{Bhaskara:2012:PIG:2095116.2095150}.
The best to-date algorithm for {\sc Densest $k$-Subgraph} due to Bhaskara et. al, is guaranteed to find a $k$-subgraph whose density is within an $\sim n^{1/4}$-multiplicative 
factor of the densest subgraph of size $k$ \cite{Bhaskara:2012:PIG:2095116.2095150}, and thus $\DKSab{\eta}{\eps}$ can be solved efficiently whenever 
$\eta \gg n^{1/4}\cdot \eps$ (this improved upon a previous $n^{1/3}$-approximation of Feige et. al \cite{FPK01}).
%There are certain pieces of evidence (see \cite{fs97}) that the above algorithms are more powerful than Linear and semidefinite programming relaxations for this problem,
%which were only shown to give an $\sim O(n/k)$-approximation for the problem (\cite{FL01}). 
Making further progress on either the lower or upper bound frontier of the 
problem is a major open problem.

Several  inapproximability results for {\sc Densest $k$-Subgraph} were known
against specific classes of algorithms \cite{Bhaskara:2012:PIG:2095116.2095150} or under assumptions that
are incomparable or stronger (thus giving weaker hardness results) than ETH: $\NP\nsubseteq\bigcap_{\epsilon>0}\BPTIME\left[2^{n^{\epsilon}}\right]$
\cite{Kho06}, Unique Games with expansion \cite{RS10}, and hardness of %random planted models and 
random $k$-CNF
\cite{Feige02, AAMMW11}. The most closely related result is by Khot \cite{Kho06}, who shows that the {\sc Densest $k$-Subgraph}
problem has no PTAS unless SAT is in {\em randomized} subexponential time. The result of \cite{Kho06}, as well as other aforementioned works, focus on the sub-constant density regime, i.e. they show hardness for distinguishing between a graph
where every $k$-subgraph is sparse, and one where every $k$-subgraph
is extremely sparse. In contrast, our result has perfect completeness and provides the first \emph{additive} 
inapproximability for {\sc Densest $k$-Subgraph} ---  the best one can hope for as per the upper 
bound of \cite{Barman14}.

The {\sc Planted Clique} problem is a special case of our problem,  where the
inputs come from a specific distribution ($G\left(n,p\right)$ versus
$G\left(n,p\right) + $ ``a planted clique of size $k$",  where $p$ is some constant,
typically $1/2$). %Notice that this problem can easily be solved in quasi-polynomial time by enumerating all potential cliques of size $10\log_2 n$.
The \emph{Planted Clique Conjecture} (\cite{AlonAKMRX07,AlonKS98, Jerrum92, Kucera95, FeigeK00, DekelGP10}) 
asserts that distinguishing between the aforementioned cases for $p=1/2, k=o(\sqrt{n})$ %indeed requires $n^{\Omega(\log n)}$ time,
cannot be done in polynomial time, and has 
served as the underlying hardness assumption in a variety of recent applications including machine-learning and cryptography (e.g. \cite{AlonAKMRX07, balcan2013finding,berthet2013complexity})
that inherently use the average-case nature of the problem, 
as well as in reductions to worst-case problems (e.g. \cite{HazanK11, AAMMW11, CLLR15-amphibious, BPRSS15-seeding}). 

The main drawback of average-case hardness assumptions is that many average-case instances (even those of worst-case-hard problems) 
are in fact tractable.
In recent years, the centrality of the planted clique conjecture inspired several works that obtain 
lower bounds in restricted models of computation  
\cite{FGRVX13-statistical, MPW15-SOS_for_planted_clique, DM15-SoS-hidden_submatrix}. 
Nevertheless, a general lower bound for the average-case planted clique problem appears out of reach for existing lower bound techniques.
Therefore,
 an important potential application of our result is replacing average-case assumptions such as the planted-clique conjecture, in applications
that do not inherently rely on the distributional nature of the inputs (e.g., when the ultimate goal is to prove a worst-case hardness result).
In such applications, there is a good chance that planted clique hardness assumptions can be replaced with a more ``conventional'' hardness assumption, such as the ETH, even when the problem has a quasi-polynomial algorithm. 
Recently, such a replacement of the planted clique conjecture with ETH was obtained for the problem of finding an approximate Nash equilibrium 
with approximately optimal social welfare \cite{BKW15}. 

We also remark that, while showing hardness for {\sc Planted Clique} from worst-case assumptions seems
beyond the reach of current techniques, our result can also be seen
as circumstantial evidence that this problem may indeed be hard. In
particular, any polynomial time algorithm (if exists) would have to inherently
use the (rich and well-understood) structure of $G\left(n,p\right)$.

%Another appealing  program which is motivated by our work, 
%is to replace known average-case hardness with worst-case
%hardness (such as ETH) in other applications such as the ones mentioned above.
%This approach has recently been proven fruitful for the problem of finding an approximate Nash equilibrium 
%with approximately optimal social welfare \cite{BKW15}. 

%We note that in our hard instance, $k=\Omega\left(n^{1-o\left(1\right)}\right)$,  for which {\sc Planted Clique} can be solved in polynomial time {[}??{]}.   

\subsubsection*{Techniques}

Our simple construction is inspired by the  ``birthday repetition'' technique which appeared recently in \cite{AIM14,BKW15, BPR15-PCP_for_PPAD}:
given a 2CSP (e.g. {\sc 3COL}), we have a vertex for each $\tilde{\Omega}\left(\sqrt{n}\right)$-tuple
of variables and assignments (respectively, {\sc 3COL} vertices
and colorings). We connect two vertices by an edge whenever their
assignments are consistent and satisfy all 2CSP constraints induced on these tuples. 
%(the birthday paradox guarantees that a constant fraction of the tuple pairs will indeed contain at least one constraint).
In the completeness case, a clique consists of choosing all the vertices
that correspond to a fixed satisfying assignment. 
In the soundness case (where the value of the 2CSP is low), 
the ``birthday paradox'' guarantees that most pairs of vertices vertices (i.e. two $\tilde{\Omega}\left(\sqrt{n}\right)$-tuples
of variables) will have a significant intersection (nonempty CSP constraints), thus 
resulting in lower densities whenever the 2CSP does not have a satisfying assignment. In the language of two-prover games, the intuition here is that  
the verifier has a ``constant chance in catching the players in a lie if thy are trying to cheat" in the game while not satisfying the CSP.
%This construction is inspired by recent applications of birthday repetition
%{[}AIM, BKW{]}.

While our construction is simple, analyzing it is intricate. The main
challenge is to rule out a ``cheating'' dense subgraph that consists
of different assignments to the same variables (inconsistent colorings
of the same vertices in {\sc 3COL}). Intuitively, this is similar
in spirit to \emph{ proving a parallel repetition theorem where the provers
can answer some questions multiple times, and completely ignore other
questions}. 
Continuing with the parallel repetition metaphor, notice that the challenge is doubled: in addition to a cheating prover correlating her answers (the standard obstacle to parallel repetition), each prover can now also correlate which questions she chooses to answer.
Our argument follows by showing that a sufficiently large
subgraph must accumulate many non-edges (violations of either 2CSP
or consistency constraints). 
%This is subtle since for a slightly smaller subgraph, this is no longer true. 
To this end we introduce an information
theoretic argument that carefully counts the entropy of choosing a
random vertex in the dense subgraph.

\subsection{Open problems}

There are several interesting open problems related to our work. We henceforth list 
four of them that are of particular interest and potential applications. %related open problems as further directions.

\paragraph{Strengthening the inapproximability factor}

Our result states that it is hard to distinguish between a graph containing a $k$-clique and a graph that does not contain a very dense ($1-\delta$) $k$-subgraph.
The latter $(1-\delta)$ seems to be a limitation of our technique.
None of the algorithms we know (including the two quasi-polynomial time algorithms mentioned above) can distinguish in polynomial time between
a graph containing a $k$-clique and a graph that does not contain even a slightly dense ($\delta$) $k$-subgraph;
for any constant $\delta >0$, and in fact even for some sub-constant values of $\delta$.
Furthermore, there is evidence \cite{AAMMW11} that this problem may indeed be hard.
This naturally leads to the following problem.
\begin{problem}[Hardness Amplification] \label{pro:amp}
Show that for every given constant $\delta > 0$, distinguishing between the following two cases is ETH-hard:
\begin{itemize}
\item There exists $S \subset V$ of size $k$ such that $\den(S) = 1$.
\item All $S \subset V$ of size $k$ have $\den(S) \leq \delta$.
\end{itemize}
\end{problem}

%In particular, Problem~\ref{pro:amp} would imply that any constant approximation to {\sc Densest $k$-Subgraph} is intractable,
%essentially matching the inapproximability results based on Unique Games with expansion~\cite{RS10}
%and planted clique~\cite{AAMMW11}. 

We remark that a similar amplification, from ``clique versus dense'' ($\den(S) = 1$ vs.\ $\den(S) = 1-\delta$) to ``clique versus sparse'' ($\den(S) = 1$ vs.\ $\den(S) = \delta$),
was shown by Alon et al. when the ``clique vs.\ dense'' instance is drawn at random according to the planted clique model~\cite{AAMMW11}.
(Unfortunately, their techniques do not seem to apply to our hard instance.)
%Let us also note that hardness for ``almost-clique versus sparse'' ($\den(S) = 1-\delta$ vs.\ $\den(S) = \delta$) was obtained in~\cite{RS10}
%from Unique Games with expansion.

An easier variant of Problem \ref{pro:amp} is to show hardness for a large gap in the imperfect completeness regime.

\begin{problem}[Hardness Amplification - imperfect completeness]\label{prob:imperfect-completeness-factor}
Show that there exist parameters $0 < \eps \ll \eta < 1$ for which 
distinguishing between the following two cases is ETH-hard:
\begin{itemize}
\item There exists $S \subset V$ of size $k$ such that $\den(S) \geq \eta$.
\item All $S \subset V$ of size $k$ have $\den(S) \leq \eps $.
\end{itemize}
\end{problem}

\noindent We note that such gaps can be obtained from average-case hardness for a random $k$-CNF~\cite{AAMMW11} 
and from Unique Games with expansion~\cite{RS10}.

\paragraph{Beyond quasi-polynomial hardness}

Another interesting challenge is to trade the perfect completeness in our main result for 
stronger notions of hardness.
%In particular, % results for {\sc Densest $k$-Subgraph} , 
%i.e. beyond $n^{\Omega(\log n)}$ running time and/or constant inapproximability.
Indeed, there are substantial evidences which suggest that 
the ``sparse vs.\ very-sparse" regime ($\DKSab{\eta}{\eps}$) is much harder to solve.
%The hard instance in~\cite{Kho06} which gives constant inapproximability assuming no randomized algorithm for \NP~in $2^{n^\eps}$-time
%is a $(d_1,d_2)$-biregular graph where both $d_1$ and $d_2$ are some constants depending on $\eps$. 
The gap instance in~\cite{Bhaskara:2012:PIG:2095116.2095150}
where all known linear and semidefinite programming techniques fail is a very sparse instance
and has integrality gap of $\Omega(n^{2/53 - \eps})$. 
In particular, every vertex has degree $n^{1/2+o(1)}$,
compared to almost linear average degree in our instance. 
Since no other algorithms succeed in this regime (even in quasi-polynomial time), 
it is natural to look for stronger lower bounds on the running time.
%and/or stronger inapproximability factors (beyond constant).
%In light of current upper bounds, one could hope to prove much stronger (potentially $2^{n^{\eps}}$)
%lower bounds for this regime. 

\begin{problem}[Trading-off perfect completeness for stronger lower bounds]\label{prob:imperfect-completeness}
Show that there exist parameters $0 < \eps < \eta \ll 1$ for which 
distinguishing between the following two cases is \NP-hard:
\begin{itemize}
\item There exists $S \subset V$ of size $k$ such that $\den(S) \geq \eta$.
\item All $S \subset V$ of size $k$ have $\den(S) \leq \eps $.
\end{itemize}
\end{problem}

\ignore{
We would also like to point out that the hard instance in~\cite{Kho06} 
%which gives constant inapproximability in $2^{\log^d n}$ assuming randomized ETH
is a $(d_1,d_2)$-biregular graph where both $d_1$ and $d_2$ are some constants depending on $\eps$, which indeed is towards the direction of~Problem~\ref{prob:imperfect-completeness}. However, there has been no work on the direction of Problem~\ref{prob:imperfect-completeness-factor}.
}

\paragraph{Finding Stable Communities} %We list the hardness of the stable communities problem as an open problem.
The problem of finding \emph{Stable Communities} is tightly related to {\sc Densest $k$-Subgraph}, and has
received recent attention in the context of social networks and learning theory \cite{arora2012finding, arora2013new, balcan2013modeling}. 

%\mnote{O: were there any other works on SC?}
%\mnote{K: Should we add more ? }

\begin{definition}[{\sc Stable Communities} \cite{balcan2013finding}]
Let $\alpha, \beta$ with $\beta < \alpha \leq 1$ be two positive parameters. Given an undirected graph, $G = (V,E)$,
$S \subset V$ is an $(\alpha, \beta)$-cluster if $S$ is :
\begin{enumerate}
\item Internally Dense: $\forall i \in S$, $|\cN(i) \cap S| \geq \alpha |S|$.
\item Externally Sparse: $\forall i \notin S$, $|\cN(i) \cap S| \leq \beta |S|$.
\end{enumerate}
\end{definition}
Currently, only planted clique based hardness is known. 
\begin{theorem}[\cite{balcan2013finding}]
For sufficiently small (constant) $\gamma$, finding a $(1,1-\gamma)$ cluster is at least as hard as {\sc Planted Clique}.
\end{theorem}

As insinuated in the introduction, we believe it is plausible and interesting to see whether the hardness assumption of the theorem above
can be replaced with ETH.

\begin{problem}[Hardness of {\sc Stable Communities}] 
Show that for some $\alpha, \beta$ with  $\beta < \alpha \leq 1$, finding an $(\alpha,\beta)$-cluster $S$ is ETH-hard.% in quasi-polynomial time.
\end{problem}

\section{Preliminaries}

Throughout the paper we use $\den(S)\in [0,1]$ to denote the density of subgraph $S$,
\[
\den(S) := \frac{\left| \big(S \times S\big) \cap E \right| }{|S \times S|}.
\]

\subsection{Information theory}

In this section, we introduce information-theoretic quantities used in this paper. For a more thorough introduction, the reader should refer to \cite{CT91}. Unless stated otherwise, all $\log$'s in this paper are base-$2$. 

\begin{definition}  
Let $\mu$ be a probability distribution on sample space $\Omega$. The \emph{Shannon entropy} (or just \emph{entropy}) of $\mu$, denoted by $H(\mu)$, is defined as $H(\mu) := \sum_{\omega \in \Omega} \mu(\omega) \log \frac{1}{\mu(\omega)}$.
\end{definition}

\begin{definition}[Binary Entropy Function]
For $p \in [0,1]$, the binary entropy function is defined as follows (with a slight abuse of notation) $H(p) := -p \log p - (1- p) \log (1-p)$. 
\end{definition}

\begin{fact}[Concavity of Binary Entropy]
\label{fact:concavity-entropy}
Let $\mu$ be a distribution on $[0,1]$, and let $p\sim \mu$. Then
$H(\E_\mu \left[p \right]) \geq \E_\mu \left[ H(p) \right]$.
\end{fact}

\noindent For a random variable $A$ we shall write $H(A)$ to denote the entropy of the induced distribution on the support of $A$. 
We use the same abuse of notation for other information-theoretic quantities appearing later in this section.

\begin{definition} The \emph{Conditional entropy} of a random variable $A$ conditioned on $B$ is defined as
\[ H(A | B) = \mathbb{E}_b (H(A | B=b)).\]
\end{definition}

\begin{fact}[Chain Rule]
$H(AB) = H(A) + H(B|A).$
\end{fact}

\begin{fact}[Conditioning Decreases Entropy]
$H(A|B) \geq H(A|BC)$.
\end{fact}

Another measure we will use (briefly) in our proof is that of \emph{Mutual Information}, which informally captures the correlation between 
two random variables. 

\begin{definition}[Conditional Mutual Information]
The \emph{mutual information} between two random variable $A$ and $B$, denoted by $I(A; B)$ is defined as
\[ I(A; B) := H(A) - H(A|B) = H(B) - H(B|A).\]
The \emph{conditional mutual information} between $A$ and $B$ given $C$, denoted by $I(A; B|C)$, is defined as
\[ I(A;B |C) := H(A|C) - H(A|BC) = H(B|C) - H(B|AC).\]
\end{definition}

The following is a well-known fact on mutual information.

\begin{fact}[Data processing inequality]
\label{fact:dataprocess}
Suppose we have the following Markov Chain:
\begin{equation*}
X \rightarrow Y \rightarrow Z
\end{equation*}
where $X \bot Z | Y$. Then $I(X;Y) \geq I(X;Z)$ or equivalently, $H(X|Y) \leq H(X|Z)$.
\end{fact}

Mutual Information is related to the following distance measure. 
\begin{definition}[Kullback-Leiber Divergence]
Given two probability distributions $\mu_1$ and $\mu_2$ on the same sample space $\Omega$  such that $(\forall \omega \in \Omega) (\mu_2(\omega)=0 \Rightarrow \mu_1(\omega)=0)$, the \emph{Kullback-Leibler Divergence} between is defined as (also known as relative entropy)
\[ \Div{\mu_1}{ \mu_2} = \sum_{\omega \in \Omega} \mu_1(\omega) \log \frac{\mu_1(\omega)}{\mu_2(\omega)}.\]
\end{definition}

\noindent The connection between the mutual information and the Kullback-Leibler divergence is provided by the following fact.

\begin{fact} 
\label{fact:divergence}
For random variables $A, B,$ and $C$ we have 
\[I(A; B|C) = \mathbb{E}_{b,c} \left[ \Div{A_{bc}}{A_{c}}\right].\]
\end{fact}

\subsection{\TCSP and the PCP Theorem}

%In an effort to make this writeup as self-contained as possible,  we henceforth introduce the necessary background and previous results 
%leading to our main reduction.  We begin with the following definition which is central to this paper:
In the \TCSP problem, we are given a graph $G=(V,E)$ on $|V|=n$ vertices, where each of the edges $(u,v)\in E$ is associated with some 
constraint function $\psi_{u,v}: A\times A \rightarrow \{0,1\}$ which specifies a set of legal ``colorings" of $u$ and $v$, from some finite alphabet $A$
($2$ in the term ``\TCSP" stands for the ``arity" of each constraint, which always involves two variables).
Let us denote by $\psi$ the entire \TCSP instance, and define by $\OPT(\psi)$ the maximum fraction of satisfied constraints in the associated graph $G$, 
over all possible assignments (colorings) of $V$.

The starting point of our reduction is the following version of the PCP theorem, which asserts that it is \NP-hard to distinguish 
a \TCSP instance whose value is $1$, and one whose value is $1 - \eta$, where $\eta$ is some small constant:

\begin{restatable}[PCP Theorem \cite{dinur2007pcp}]{theorem}{thmTCSP}
\label{thm_MR}
Given a \TSAT instance $\varphi$ of size $n$, there is a polynomial time reduction that produces a \TCSP 
instance $\psi$, with size $|\psi| = n \cdot \polylog n$ variables and constraints, and constant alphabet size such that
\begin{itemize}
\item (Completeness) If $\OPT( \varphi ) = 1$ then $\OPT( \psi ) = 1$.
\item (Soundness) If $\OPT( \varphi ) < 1 $ then $ \OPT(\psi) < 1 - \eta$, for some constant $\eta = \Omega(1)$
\item (Balance) Every vertex in $\psi$ has degree $d$ for some constant $d$.
\end{itemize}
\end{restatable}
In the appendix, we describe in detail how to derive this formulation of the PCP Theorem from that of e.g. \cite{AIM14}.% \ref{thm_MR} in the appendix.

Notice that since the size of the reduction is near linear, ETH implies that solving the above problem requires near exponential time.

\begin{corollary}
Let $\psi$ be as in Theorem \ref{thm_MR}. Then assuming ETH, distinguishing between $\OPT( \psi ) = 1$ and $ \OPT(\psi) < 1 - \eta$ requires time $2^{\tilde \Omega\left(|\psi|\right)}$.
\end{corollary}
% !TeX root = Dks_main.tex
\section{Main Proof}

\subsection{Construction}

%Let $\eta$ be a sufficiently small constant to be defined later.
Let $\psi$ be the \TCSP instance produced by the reduction in Theorem \ref{thm_MR}, % setting $\delta=\eta$, 
i.e. a constraint graph over $n$ variables  with alphabet $A$ of constant size.
We construct the following graph $G_\psi = (V, E)$:
\begin{itemize} 
\item Let  $\rho := \sqrt{n} \log \log n$ and $k := {{n}\choose{\rho}}$.
\item Vertices of $G_\psi$ correspond to all possible assignments (colorings) to all $\rho$-tuples of variables in  $\psi$, i.e $V= [n]^\rho \times |A|^\rho$.
% colorings of 
%all the tuples in $$.
Each vertex is of the form $v = (y_{x_1},y_{x_2},\ldots ,y_{x_\rho})$ where $\{x_1,\ldots , x_\rho\}$ are the chosen variables of $v$, and 
$y_{x_i}$ is the corresponding assignment to variable $x_i$.
\item If $v \in V$ {\em violates any \TCSP constraints}, i.e. if there is a constraint on $\left(x_i,x_j\right)$ in $\psi$ which is not satisfied by $\left(y_{x_i},y_{x_j}\right)$, then $v$ is an isolated vertex in $G_{\psi}$.% (i.e., it is disconnected).
\item Let $u = (y_{x_1},y_{x_2},\ldots ,y_{x_\rho})$ and $v = (y'_{x'_1},y'_{x'_2},\ldots ,y'_{x'_\rho})$.
$(u,v) \in E$ iff: 
	\begin{itemize} 
		\item $(u,v)$ does not {\em violate any consistency constraints}: for every shared variable $x_i$, the corresponding assignments agree, $y_{x_i} = y'_{x_i}$;
		\item and $(u,v)$ also does not {\em violate any \TCSP constraints}: for every \TCSP constraint on $\left(x_i,x'_j\right)$ (if exists), 
				the assignment $\left(y_{x_i},y'_{x'_j}\right)$ satisfy the constraint.
	\end{itemize}
\end{itemize} 
Notice that the size of our reduction (number of vertices of $G_{\psi}$) is $N = {n \choose \rho} \cdot |A|^\rho = 2^{\Otilde( \sqrt{n} )}$. 

%By the premise of the theorem, $\psi$ is either completely satisfiable (``perfect completeness") or at most $1-\eta$ fraction of the constrains (edges) are satisfied.

\paragraph{Completeness} If $\OPT(\psi) = 1$, then $G_\psi$ has a $k$-clique:
Fix a satisfying assignment for $\psi$, and let $S$ be the set of all vertices that are consistent with this assignment. 
Notice that $|S|={n\choose{\rho}}=k$. Furthermore its vertices do not violate any consistency constraints (since they agree with a single assignment), or \TCSP constraints (since we started from a satisfying assignment).

% !TeX root = Dks_main.tex

\section{Soundness} 

Suppose that $\OPT(\psi) < 1 - \eta$, and let $\epsilon_0 >0$ be some constant to be determined later. We shall show that for any subset $S$ of size 
$k'=k\cdot |V|^{-\epsilon_0 / \log\log |V|}$, $\den(S) < 1 - \delta$, where $\delta$ is some constant depending on 
$\eta$. %For the remainder of this section, we assume that $\OPT(H) < 1 - \eta$ and devote to proving the following theorem;
The remainder of this section is devoted to proving the following theorem:

\begin{restatable}[Soundness]{theorem}{thmSoundness}
\label{thm:soundness}
If $\OPT(\psi) < 1 - \eta$, then $\forall S \subset V$ of size $k'=k \cdot |V|^{-\epsilon_0 / \log\log |V|}$, $\den(S) < 1 - \delta$ for some constant $\delta$.
\end{restatable}

\subsection{Setting up the entropy argument}

Fix some subset $S$ of size $k'$, and let $v \in_R S$ be a uniformly chosen vertex in $S$ (recall that $v$ is a vector of $\rho$ coordinates, corresponding
to labels for a subset of $\rho$ chosen variables).
Let $X_i$ denote the indicator variable associated with $v$ such that $X_i=1$ if the $i$'th variable appears in $v$ and $0$ otherwise.
We let $Y_i$ represent the coloring assignment (label) for the $i$'th variable whenever $X_i = 1$,  which is of the form $l \in A$. 
Throughout the proof, let  $$W_{i-1} = X_{< i}, Y_{< i}$$
denote the $i$'th prefix corresponding to $v$. We can write : 
\begin{align*}
H(Y_i | W_{i -1}, X_i) %&~=~ \sum_{w_{i-1}} \Pr[ X_i = 1, W_{i-1} = w_{i-1} ] \cdot H(Y_i | W_{i-1} = w_{i-1}, X_i = 1) \\
%&~+~ \sum_{w_{i-1}} \Pr[ X_i = 0, W_{i-1} = w_{i-1}  ] \cdot \underbrace{H( Y_i | X_i = 0 )}_{ = 0} \\
%&~=~ \sum_{w_{i-1}} \Pr[ X_i = 1 ] \Pr[W_{i-1} = w_{i-1} | X_i = 1] \cdot H(Y_i | W_{i-1} = w_{i-1}, X_i = 1) \\
%&~=~ \Pr[ X_i = 1] \cdot \sum_{w_{i-1}} \Pr[W_{i-1} = w_{i-1} | X_i = 1] \cdot H(Y_i | W_{i-1} = w_{i-1}, X_i = 1) \\
&~=~ \Pr[X_i = 0] \cdot H(Y_i | W_{i-1}, X_i = 0) + \Pr[X_i = 1] \cdot H(Y_i | W_{i-1}, X_i = 1) \\
&~=~ \Pr[X_i = 1] \cdot H(Y_i | W_{i-1}, X_i = 1) 
\end{align*}
since $H( Y_i | W_{i -1}, X_i = 0 ) = 0$.
Notice that since $(XY)$ and $v$ determine each other, and $v$ was uniform on a set of size $|S|=k'$, we have 
\begin{observation}
$H(XY) = \log k'$.
%Total entropy on the distribution induced by $k$-sized subset should be $\log k$, since it is a uniform distribution over $k$ elements. In particular, 
%using the chain rule for entropy, 
%\begin{align*}
%& H(XY)~=~\sum_i H(X_i | X_{< i}, Y_{< i}) + H(Y_i | X_{\leq i}, Y_{< i}) \\
%&~=~  \sum_i H(X_i | X_{ < i }, Y_{ < i }) + \Pr[X_i = 1] \cdot H( Y_i | X_{ < i }, Y_{ < i }, X_i = 1) \\
%&~=~ \log k
%\end{align*}
\end{observation}

Thus, in total, the choice of challenge and the choice of assignments should contribute $\log k'$ to the entropy of $v$. If much 
of the entropy comes from the assignment distribution (conditioned on the fixed challenge variables), we will show that $S$ must have many consistency violations, implying that $S$ is sparse. If, on the other hand, almost all the entropy comes from the challenge distribution, we will show that this implies many CSP constraint
violations (implied by the soundness assumption). 
From now on, we denote $$\alpha_i := H(X_i | X_{< i}, Y_{< i}) \;\;\; \text{and}  \;\;\; \beta_i := H(Y_i | X_{\leq i}, Y_{< i}).$$ 
When conditioning on the $i$'th prefix, we
shall write $\alpha_i (w_{i-1}) := H(X_i | X_{< i}, Y_{< i} = w_{i-1})$, 
and similarly for $\beta_i( \cdot )$. Also for brevity, we denote $$q_i := \Pr[ X_i = 1] \;\;\; \text{and}  \;\;\;  q_i ( w_{i-1} ) := \Pr [X_i = 1 | w_{i-1} ]. $$

\subsubsection*{Prefix graphs}

The consistency constraints induce, for each $i$, a graph over the prefixes:
the vertices are the prefixes, and two prefixes are connected by an edge if their labels are consistent.
(We can ignore the \TCSP constraints for now --- the prefix graph will be used only in the analysis of the consistency constraints.)
Formally,

\begin{definition}[Prefix graph] 
For $i \in [n+1]$ let the {\em $i$-th prefix graph}, $G_i$ be defined over the prefixes of length $i-1$ as follows.
We say that $w_{i-1}$ is a neighbor of $\sigma_{i-1}$ if they do not violate any consistency constraints. 
Namely, for all $j < i$, if $X_j = 1$ for both $w_{i-1}$ and $\sigma_{i-1}$, then $w_i$ and $\sigma_i$ assign the same label $Y_j$. 

In particular, we will heavily use the following notation: let $\cN( w_{i-1} )$ be the {\em prefix neighborhood} of $w_{i-1}$; i.e. it is the set of all prefixes (of length $i-1$) that are consistent with $w_{i-1}$.
For technical issues of normalization, we let $w_{i-1} \in \cN( w_{i-1} )$, i.e. all the prefixes have self-loops.
\end{definition}

Notice that $G_{n+1}$ is defined over the vertices of $S$  (the original subgraph). 
The set of edges on $S$ is contained in the set of edges of $G_{n+1}$, 
since in the latter we only remove pairs that violated consistency constraints 
(recall that we ignore the \TCSP constraints).

Unless stated otherwise, we always think of prefixes as weighted by their probabilities. 
Naturally, we also define the weighted degree and weighted edge density of the prefix graph.
%There is a small technicality: for the normalization to work out, we need to add a self-loop of weight $1/2$ on each prefix.

\begin{definition}[Prefix degree and density]
The {\em prefix degree} of $w_{i-1}$ is given by:
\begin{equation*}
\deg(w_{i-1}) = \sum_{\sigma_{i-1} \in \cN( w_{i-1} )} \Pr[\sigma_{i-1}].
\end{equation*}
Similarly, we define the {\em prefix density} of $G_i$ as:
\begin{equation*}
\den (G_i) = \sum_{w_{i-1}} \sum_{\sigma_{i-1}  \in \cN( w_{i-1} )} \Pr[w_{i-1}] \cdot \Pr[\sigma_{i-1}] .
\end{equation*}
\end{definition}

When it is clear from the context, we henceforth drop the {\em prefix} qualification, and simply refer to the {\em neighborhood} or {\em degree}, etc., of $w_{i-1}$.

Notice that in $G_{n+1}$, the probabilities are uniformly distributed.
In particular, $\den(G_{n+1} ) \geq \den(S)$, since, as we mentioned earlier, the set of edges in $S$ is contained in that of $G_{n+1}$.
Finally, observe also that because we accumulate violations, the density of the prefix graphs is monotonically non-increasing with $i$.

\begin{observation}
\label{obs:density-monotone}
\begin{equation*} \den(G_1) \geq \dots \geq \den(G_{n+1} ) \geq \den(S). \end{equation*}
\end{observation}

\subsubsection*{Useful approximations}
We use the following bounds on $\alpha_i$ and $\beta_i$ many times throughout the proof:

\begin{fact} 
\label{fact:alphabound}
\begin{equation*}
\alpha_i = \E \left[ H( q_i (w_{i-1}) ) \right] \leq H( \E \left[ q_i (w_{i-1}) \right] ) = H(q_i)
\end{equation*}
\end{fact}
\begin{fact}
\label{fact:betabound}
\begin{equation*} 
\beta_i  = \E \left[ \beta_i (w_{i-1}) ) \right] \leq \E \left[ q_i (w_{i-1}) \cdot \log |A|\right]   = q_i \log |A| 
\end{equation*}
\end{fact}
\begin{proof} 
The bound on $\alpha_i$ follows from concavity of entropy (Fact \ref{fact:concavity-entropy}).
For the second bound, observe that $\beta_i$ is maximized by spreading $q_i$ mass uniformly over alphabet $A$. 
\end{proof}

We also recall some elementary approximations to logarithms and entropies
that will be useful in the analysis.  The proofs are deferred to the appendix.

\begin{fact} \label{fact:logk-expansion}
For $k = {n \choose \rho}$ then,
\begin{equation*}
\log k=nH\left(\frac{\rho}{n}\right)\pm O\left(\log n\right)=\left(\frac{1}{2}-o\left(1\right)\right)\rho\log n
\end{equation*}
\end{fact}

More useful to us will be the following bounds on $\log k'$:

\begin{fact}
\label{fact:logkprime-expansionNEW}
Let $\epsilon_1 \ge 5\epsilon_0$, and $k,k',V,n,\rho$ as specified in the construction. Then,
\begin{gather*}
\log k'  \geq \max \left\{\log k, n H\left(\frac{\rho}{n}\right) \right\} -\underbrace{\epsilon_1 \log k/\log n}_{\approx \frac{\eps_1}{2} \cdot \rho}.
\end{gather*}
%\begin{gather*}
%\log k'  \geq \max \left\{\log k, nH\left(\frac{\rho}{n}\right), \frac{1}{3}\rho\log n \right\} -\underbrace{\epsilon_1 \log k/\log n}_{\approx \frac{\eps_1}{2} \cdot \rho}.
%\end{gather*}

In particular, this means that most indices $i$ %(except the last few), we expect to
should contribute roughly $H\left(\frac{\rho}{n}\right)$ entropy to the choice of $v$.
\end{fact}

We will also need the following bound which relates the entropies of a very biased coin and a slightly less biased one:
%, for small deviation on binary entropy function from $1/n$. 
\begin{fact}
\label{fact:entropy-taylor}Let $1/n\ll\left|\upsilon\right|\ll1$,
then 
\[
H\left(\frac{1+\upsilon}{n}\right)=H\left(\frac{1}{n}\right)-\frac{\upsilon}{n}\log\frac{1}{n}-\left(\log\e\right)\frac{\upsilon^{2}}{2n}+O\left(n^{-2}\right)+O\left(\frac{\upsilon^{3}}{n}\right)
\]
\end{fact}

\subsection{Consistency violations}
In this section, we show that if the entropy contribution of the assignments ($\sum_i H( Y_i | X_{ \leq i }, Y_{ < i })$) is large, there are many consistency 
violations between vertices, which lead to constant density loss. First, we show that if $H( Y_i | X_{ \leq i }, Y_{ < i }) > 5 \eps_1 \log k / \log n$, 
then at least a constant fraction of such entropy is concentrated on ``good'' variables.

\begin{definition}[Good Variables]
We say that an index $i$ %(corresponding to $x_i$) 
is {\em good} if
\begin{itemize} 
\item $\alpha_i \geq H( q_i ) - 2 q_i \log |A|$
\item $\beta_i \geq \frac{1}{2} \eps_1  q_i $ 
\end{itemize}
where $\eps_1$ is a constant to be determined later in the proof.
\end{definition}
%%%%%%%%%%%%%%%%%%%%%%%%
\begin{claim} \label{cla:goodvar} For any constant $\eps_1$, if $\sum_i \beta_i > 5\eps_1 \log k / \log n$,
\[
\sum_{\mbox{good \ensuremath{i}'s}} q_i^2  \geq 
\left(\frac{1}{5} \eps_{1} \rho\right)^2 / (n \log^2 |A|) = \Omega(\rho^2 /n) .
\]
\end{claim}
\begin{proof}
%For any fixed $q_i$, $ 
%Recall that $q_i = \E_{w_{i-1}}[\alpha_i] = \sum_{w_{i-1}} \Pr[w_{i-1}]\cdot H(X_i|w_{i-1})$, thus 

%\begin{align*}
%\sum_{i\in\left[n\right]} \alpha_i &= \log k-\sum_{i\in\left[n\right]}\beta_{i}\\
%& \geq\sum_{i\in\left(n\right)}\left(H\left(q_{i}\right)-\beta_{i}\right) -O(\log n) \\
%& \geq\sum_{i\in\left(n\right)}\left(H\left(q_{i}\right)-q_{i}\log|A|\right) -O(\log n),
%\end{align*}
%where the first inequality follows from $\log k=nH\left(\frac{\rho}{n}\right)-O(\kog n)$
%and concavity of entropy, and the second inequality from $\beta_i \leq q_{i}\log|A|$.

%So both $\sum\beta_{i}$ and $\sum\alpha_{i}$ are not too small.  
%\mnote{explain last equation in p.7. also, wording of this proof is very vague, i can't understand it. please rephrase }
We want to show that many of the indices $i$ have both a large $\alpha_{i}$
and a large $\beta_{i}$ simultaneously. We can write
\begin{gather*}
\sum_{i\in\left[n\right]}\left(\alpha_{i}+\beta_{i}\right)  =  \sum_{i\colon\alpha_{i}+\beta_{i}<H\left(q_{i}\right)-q_{i}\log|A|}\left(\alpha_{i}+\beta_{i}\right)+\sum_{i\colon\alpha_{i}+\beta_{i}\geq H\left(q_{i}\right)-q_{i}\log|A|}\left(\alpha_{i}+\beta_{i}\right)
\end{gather*}
Using Facts \ref{fact:alphabound} and \ref{fact:betabound},
we have
\begin{gather}
\sum_{i\in\left[n\right]}\left(\alpha_{i}+\beta_{i}\right)\leq\sum_{i\colon\alpha_{i}+\beta_{i}<H\left(q_{i}\right)-q_{i}\log|A|}\left(H\left(q_{i}\right)-\beta_{i}\right)+\sum_{i\colon\alpha_{i}+\beta_{i}\geq H\left(q_{i}\right)-q_{i}\log|A|}\left(H\left(q_{i}\right)+\beta_{i}\right)\label{eq:sum(a_i+b_i)}.
\end{gather}
Because the subgraph is of size $k'$, from the expansion of $\log k'$ (Fact \ref{fact:logkprime-expansionNEW}), 
\begin{gather*} 
\sum_{i\in\left[n\right]}\left(\alpha_{i}+\beta_{i}\right)\ge n H\left(\frac{\rho}{n}\right)-\epsilon_1 \log k /\log n \geq \sum H\left(q_{i}\right)-\epsilon_1 \log k /\log n,
\end{gather*}
where the second inequality follows from the concavity of entropy. 
Plugging into \eqref{eq:sum(a_i+b_i)}, we have
\begin{align*}
 \sum_{i\colon\alpha_{i}+\beta_{i}\geq H\left(q_{i}\right)-q_{i}\log|A|}\beta_{i} &\geq \sum_{i\colon\alpha_{i}+\beta_{i}<H\left(q_{i}\right)-q_{i}\log|A|}\beta_{i}-\epsilon_1 \log k /\log n\\
& = \left( \sum_i \beta_i - \sum_{i\colon\alpha_{i}+\beta_{i}\geq H\left(q_{i}\right)-q_{i}\log|A|}\beta_{i} \right) -\epsilon_1 \log k /\log n
\end{align*}
Rearranging, we get
\begin{equation}
\sum_{i\colon\alpha_{i}+\beta_{i}\geq H\left(q_{i}\right)-q_{i}\log|A|}\beta_{i} \geq \frac{1}{2} \sum_i \beta_i -\epsilon_1 \log k /\log n
\label{eq:beta_i-good-mass2}
\end{equation}

For all the $i$'s in the LHS summation, $\alpha_{i}\geq H\left(q_{i}\right)-2q_{i}\log|A|$ by Fact \ref{fact:betabound}.
From now on, we will consider only $i$'s that satisfy this condition. Now, using the premise on $\sum_i \beta_i$ and \pref{eq:beta_i-good-mass2} we have: 

\begin{gather*}
\sum_{i\colon\alpha_{i}\geq H\left(q_{i}\right)-2q_{i}\log|A|}\beta_{i}
\geq (5/2-1)\epsilon_1 \log k /\log n
\geq 0.7 \eps_{1}\rho,
\end{gather*}
where the second inequality follows from our approximation for $\log k$ (Fact \ref{fact:logk-expansion}).

We want to further restrict our attention to $i$'s for which $\beta_{i}$
is at least $\frac{1}{2}\eps_{1}q_i$ (aka good $i$'s).
Note that the above inequality can be decomposed to 
\begin{equation*} 
\sum_{\mbox{good \ensuremath{i}'s}}\beta_{i} +\sum_{\substack{i\colon\alpha_{i}\geq H\left(q_{i}\right)-2q_{i}\log|A| \\ \beta_i <\frac{1}{2}\eps_{1}q_i}}\beta_{i} \geq 0.7\eps_{1}\rho
\end{equation*}
Now via a simple sum bound,
\begin{equation*}
\sum_{\substack{i\colon\alpha_{i}\geq H\left(q_{i}\right)-2q_{i}\log|A| \\ \beta_i <\frac{1}{2}\eps_{1}q_i}}\beta_{i}   \leq
\frac{1}{2} \eps_{1} \sum_i  q_i = 
 \frac{1}{2}\eps_{1}\rho
\end{equation*}
Rearranging, we get, 
\[
\sum_{\mbox{good \ensuremath{i}'s}}\beta_{i}\geq \frac{1}{5} \eps_{1} \rho
\]
By Cauchy-Schwartz we have:
\[
\sum_{\mbox{good \ensuremath{i}'s}}\beta_{i}^2\geq\left(\frac{1}{5} \eps_{1} \rho \right)^2/n
\]
Finally, since $\beta_i \leq q_i \log |A|$,
\[
\sum_{\mbox{good \ensuremath{i}'s}}q_i^2 \geq\left(\frac{1}{5} \eps_{1} \rho \right)^2/(n \log ^2 |A|).
\]

\end{proof}

In the same spirit, we now define a notion of a ``good'' prefix. 
Intuitively, conditioning on a good prefix leaves a significant amount of entropy on 
the $i$'th index. 
We also require that a good prefix has a high prefix degree;
that is, it has many neighbors it could potentially lose when revealing the $i$-th label. 
%We begin by formally defining our notion of consistent prefixes $w_i$ and $\sigma_i$, and the ``consistent neighborhood" of a prefix $w_i$:
%We are ready to define ``good'' prefixes.
\begin{definition}[Good Prefixes]
\label{def:good_prefix}
We say $w_{i-1}$ is a {\em good} prefix if
\begin{itemize} 
\item $i$ is good; 
\item $\sum_{ \sigma_{i-1} \in \cN ( w_{i-1} ) } q_{i}(\sigma_{i-1}) \Pr[ \sigma_{i-1} ] \geq ( 1 - \eps_2) q_{i}$;
\item $\beta_{i} (w_{i-1})  \geq \eps_3 q_{i}(w_{i-1})$ 
\end{itemize}
where $\eps_3 =  (\eps_4 + \kappa) \log |A|$, with $\eps_4$ an an arbitrarily small constant that denotes the fraction of assignments that disagree with the majority of the assignments, $\kappa = \Theta( 1/ \log |A|)$ factor, and $\eps_2$ a constant that satisfies $\delta = \left( \frac{\eps_2}{|A|^{2/\eps_2}} \right)^4$, with $\den(S) = 1 - \delta$.
\end{definition}

In the following claim, we show that these prefixes contribute some constant fraction of entropy, assuming that our subset is dense.

\begin{claim}\label{cla:goodpre} If $\den(S) > 1 - \delta$, where $\delta = \left( \frac{\eps_2}{|A|^{2/\eps_2}} \right)^4$ and $\eps_1 \geq 4 \eps_2 \log |A| + 8 \eps_3$, then
for every good index $i$, it holds that
\[
\sum_{\mbox{good \ensuremath{w_{i-1}}'s}}
\Pr[w_{i-1}] \beta_{i}\left(w_{i-1}\right)
\geq \beta_i / 4
\]
\end{claim}
\begin{proof}
We begin by proving that most prefixes satisfy the degree condition of Definition \ref{def:good_prefix}.
Let $w_{i-1}$ be {\em popular} if $i$ is a good variable
and its degree in the prefix graph $G_i$ is at least $deg(w_{i-1}):=\sum_{\sigma_{i-1}\in{\cal N}\left(w_{i-1}\right)}\Pr[ \sigma_{i-1} ]$ $ \geq1-\sqrt{\delta}$.
Recall that $\den(G_i) \geq \den(S) \geq (1- \delta)$ (by Observation \ref{obs:density-monotone}).
Thus by Markov inequality, at most $\sqrt{\delta}$-fraction of the prefixes are unpopular. 

Let $Z(\cdot)$ be the indicator variable for $W_{i-1}$ being popular.
For the sake of contradiction, suppose that more than $\eps_2$-fraction of the
$q_i$-mass is concentrated on unpopular prefixes, that is:
\begin{gather} \label{eq_unpopular_fraction}
\sum_{\mbox{unpopular \ensuremath{w_{i-1}}'s}}\Pr[ w_{i-1} ]q_{i}\left(w_{i-1}\right) = \Pr\left[Z(W_{i-1})=0\right]\cdot \Pr\left[X_i=1\; | \; Z(W_{i-1})=0\right]  > \eps_{2}q_{i}.
\end{gather}
We would like to argue that this condition implies that the distribution on the $X_i$'s is highly biased by the conditioning on the (popularity of the) 
prefix;  this in turn implies that $\alpha_i$, the expected conditional entropy of $X_i$, must be low, contradicting the assumption that $i$ is good. 
Indeed, by data-processing inequality (Fact \ref{fact:dataprocess}),
\begin{eqnarray}
\alpha_{i} & = & H\left(X_i \mid W_{i-1}\right)\nonumber \\
 & \leq & H\left(X_i \mid Z(W_{i-1}) \right)\nonumber \\
 & = & H\left( X_i \right)-I\left(X_i ; Z(W_{i-1})\right)\label{eq:I(X;Y)}
\end{eqnarray}
Since we can write mutual information as expected KL-divergence (Fact \ref{fact:divergence}), and KL-divergence is non-negative, we get 
\begin{align*}
 I( X_i ; Z(W_{i-1}) ) & = \E_{x_i}\left[ \Div{Z(W_{i-1})|x_i}{Z(W_{i-1})} \right]  \\ 
&\geq q_i \cdot \Div{\Pr[Z(W_{i-1}) =1 \; | x_i =1]}{Z(W_{i-1}) =1}\\ 
& \geq q_i \cdot \Div{ 1- \eps_2 }{ 1- \sqrt{\delta}} = q_i \Div{\eps_2}{\sqrt{\delta}},
\end{align*}
where the second inequality follows from the fact that for all good $i$'s, our degree assumption implies 
$\Pr[Z(W_{i-1})] \geq (1-\sqrt{\delta})$, and  our assumption in \eqref{eq_unpopular_fraction}
implies, via Bayes rule, that $\Pr[Z(W_{i-1}=0\; | \; x_i = 1] \geq \eps_2$, and therefore
$\Pr[W_{i-1} =1 \; | \; x_i = 1] \leq 1 - \eps_2$. Note that by our setting of parameters 
$1-\sqrt{\delta}>1-\varepsilon_2$. 

Plugging into (\ref{eq:I(X;Y)}) we have:
\begin{equation}
\alpha_{i}\leq H\left(q_i \right) - q_{i}\Div{\eps_{2}}{\sqrt{\delta}}\mbox{.}\label{eq:KL-div}
\end{equation}
On the other hand, recall that since $i$ is good, $\alpha_{i}\geq H\left(q_{i}\right)-2q_{i}\log|A|$.
Recall also that $\delta = \left( \frac{\eps_2}{|A|^{2/\eps_2}}\right)^4$, and therefore $\Div{\eps_2}{ \sqrt{\delta}}\geq2\log|A|$. Thus, we get a contradiction to \eqref{eq_unpopular_fraction}.
From now on we assume
\begin{equation}
\sum_{\mbox{unpopular \ensuremath{w_{i-1}}'s}}\Pr[w_{i-1}]q_{i}\left(w_{i-1}\right)\leq\eps_{2}q_{i}.      \label{eq:q_i-concentrated-on-neighbors-1}
\end{equation}
%Recalling that $\beta_i \leq q_i \log |A|$, the above 
This implies that even if the assignment is uniform over the alphabet, the contribution to $\sum\beta_{i}$  
from unpopular prefixes is small:
\begin{align*}
\sum_{\mbox{unpopular \ensuremath{w_{i-1}}'s}}\Pr[w_{i-1}]\beta_{i}\left(w_{i-1}\right) &~\leq~ \sum_{\mbox{unpopular \ensuremath{w_{i-1}}'s}}\Pr[w_{i-1}] q_{i}\left(w_{i-1}\right) \log |A| \\
&~\leq~ \eps_{2}q_{i}\log|A|
~\leq~ \frac{1}{4}\eps_1 q_i ~\leq~ \frac{1}{2} \beta_i
\end{align*}
where first inequality follows from Fact \ref{fact:betabound}, second from \pref{eq:q_i-concentrated-on-neighbors-1}, third from our setting of $\eps_1 \geq 4 \eps_2 \log |A|$, and fourth from $\beta_{i}\geq\frac{1}{2}\eps_{1}q_{i}$ since $i$ is good. Therefore,
\begin{gather*}
\sum_{\mbox{popular \ensuremath{w_{i-1}}'s}}
\Pr[w_{i-1}]\beta_{i}\left(w_{i-1}\right)  = 
\beta_i -
\sum_{\mbox{unpopular \ensuremath{w_{i-1}}'s}}
\Pr[w_{i-1}]\beta_{i}\left(w_{i-1}\right)  \geq
\beta_i/2
\end{gather*}

Using a similar argument, we show that for any popular $w_{i-1}$,
most of the $q_{i}$ mass is concentrated on its neighbors. 
Consider any popular $w_{i-1}$, and let ${\cal N}^{C}\left(w_{i-1}\right)$ denote the complement of ${\cal N}\left(w_{i-1}\right)$.
Then we can rewrite $\alpha_{i}$ as:
\begin{eqnarray*}
\alpha_{i} & = & \sum_{\sigma_{i-1}\in{\cal N}\left(w_{i-1}\right)}\Pr[ \sigma_{i-1} ]\alpha_{i}\left(\sigma_{i-1}\right)+\sum_{\sigma_{i-1}\in{\cal N}^{C}\left(w_{i-1}\right)}\Pr[ \sigma_{i-1} ]\alpha_{i}\left(\sigma_{i-1}\right)
\end{eqnarray*}
Notice that since $w_{i-1}$ is popular, ${\cal N}^{C}\left(w_{i-1}\right)$
has measure at most $\sqrt{\delta}$. Thus, if an $\eps_{2}$-fraction
of the $q_{i}$ mass is concentrated on ${\cal N}^{C}\left(w_{i-1}\right)$,
we once again (like in (\ref{eq:KL-div})) have 
\[
\alpha_{i}\leq H\left(q_{i}\right)-q_{i}\Div{\eps_{2}}{\sqrt{\delta}}\mbox{,}
\]
which would again yield a contradiction to $i$ being a good variable.
Therefore every popular prefix also satisfies the $q_{i}$-weighted
condition on the degree: 
\begin{equation}
\sum_{\sigma_{i-1}\in{\cal N}\left(w_{i-1}\right)}\Pr[ \sigma_{i-1} ]q_{i}\left(\sigma_{i-1}\right)\geq\left(1-\eps_{2}\right)q_{i}\label{eq:q_i-concentrated-on-neighbors}
\end{equation}

Recall that a prefix $w_{i-1}$ is good if it also satisfies
$\beta_{i}\left(w_{i-1}\right)\geq\eps_{3}\cdot q_{i}\left(w_{i-1}\right)$.
Fortunately, prefixes that violate this condition (i.e. those with small $\beta_{i}\left(w_{i-1}\right)$), 
cannot account for much of the weight on $\beta_{i}$:
\[
\sum_{\beta_{i}\left(w_{i-1}\right)<\eps_{3}q_{i}\left(w_{i-1}\right)}
\Pr[w_{i-1}]\beta_{i}\left(w_{i-1}\right) \leq
\eps_{3} q_{i} .
\]
Since $i$ is good and $\eps_{1}\geq 8 \eps_{3}$, this implies:
\begin{gather*}
\sum_{\mbox{good \ensuremath{w_{i-1}}'s}}
\Pr[w_{i-1}] \beta_{i}\left(w_{i-1}\right)
\geq \beta_i / 2 - \eps_3 q_i
\geq \beta_i / 4
\end{gather*}
since 
\begin{equation*}
\eps_3 q_i \leq \frac{1}{8} \eps_1 q_i \leq \frac{1}{4} \beta_i 
\end{equation*}
where last inequality follows from $i$ being good.
\end{proof}

\begin{corollary}
\label{cor:good_prefix}
For every good index $i$,
\[
\sum_{\mbox{good \ensuremath{w_{i-1}}'s}}
\Pr[w_{i-1}] q_{i}\left(w_{i-1}\right)
\geq \frac{\eps_1 }{8 \log |A|}q_i.
\]
\end{corollary}
\begin{proof}
\begin{align*}
\sum_{\mbox{good \ensuremath{w_{i-1}}'s}}
\Pr[w_{i-1}] q_{i}\left(w_{i-1}\right) & \geq 
\sum_{\mbox{good \ensuremath{w_{i-1}}'s}}
\Pr[w_{i-1}] \beta_{i} / \log|A|  && \text{(Fact \ref{fact:betabound})}\\
& \geq \beta_i / (4 \log |A|)  && \text{(Claim \ref{cla:goodpre})} \\
& \geq \frac{\eps_1 }{8 \log |A|}q_i && \text{(Definition of good $i$)}
\end{align*}
\end{proof}
%%%%%%%%%%%%%%

With~\pref{cla:goodvar} and Corollary \pref{cor:good_prefix}, we are ready to prove the main lemma of this section: 
\begin{lemma}[Labeling Entropy Bound] \label{lem:consistency}
If $\sum_i H( Y_i | X_{ \leq i }, Y_{ < i }) > \frac{5\eps_1 \log k}{\log n}$, then $\den(S) < 1 - \delta$. 
\end{lemma}
\begin{proof}
Assume for a contradiction that $\den(S) \geq 1- \delta$.
For prefix $w_{i-1}$, let $\cD_{w_{i-1}}$ denote the induced distribution on labels to the $i$-th variable, conditioned on $w_{i-1}$ and $x_i=1$.
(If $q_i (w_{i-1})=0$, take an arbitrary distribution.) 
%Observe that we can write $\frac{\left|\left(S\times S\right)\setminus E\right|}{\left|S\times S\right|} = 1 - \den(S)$ as 
After revealing each variable $i$, the loss in prefix density is given by the probability of ``fresh violations": the sum over all prefix edges $(w_{i-1},\sigma_{i-1})$
of the probability that they assign different labels to the $i$-th variable:
\begin{equation}
\label{eq:density_loss}
\den(G_i) - \den(G_{i+1}) = \sum_{w_{i-1}} \sum_{\sigma_{i-1} \in \cN (w_{i-1})} \Big(\Pr[w_{i-1}] \Pr[\sigma_{i-1}] q_{i}(w_{i-1})  q_{i}(\sigma_{i-1})\Big)
\Pr_{\substack{Y_i \sim \cD_{w_{i-1}}\\ Y'_i \sim \cD_{\sigma_{i-1}}}} [Y_i \neq Y'_i]
\end{equation}

%\begin{align*}
%\frac{\left|\left(S\times S\right)\setminus E\right|}{\left|S\times S\right|} &~=~ \sum_i \sum_{w_{i-1}} \sum_{\sigma_{i-1} \in \cN (w_{i-1})} \frac{1 - \ip{\cD_{w_{i-1}}, \cD_{\sigma_{i-1}}}}{4} \Pr[w_{i-1}] \Pr[\sigma_{i-1}] q_{i}(w_{i-1})  q_{i}(\sigma_{i-1})\\
%&~\geq~ \sum_{\substack{\mbox{good \ensuremath{i}'s} \\ \mbox{good \ensuremath{w_{i-1}}'s}}} \sum_{\sigma_{i-1} \in \cN (w_{i-1})} \frac{1 - \ip{\cD_{w_{i-1}}, \cD_{\sigma_{i-1}}}}{4} \Pr[w_{i-1}] \Pr[\sigma_{i-1}] q_{i}(w_{i-1})  q_{i}(\sigma_{i-1})
%\end{align*}
%where $\cD_{w_{i-1}}$ refers to the distribution of $Y_{i}$ conditioned on prefix being $w_{i-1}$ and $X_{i} = 1$, similarly for $\cD_{\sigma_{i-1}}$. 

We now lower-bound $\Pr_{\cD_{w_{i-1}} \times \cD_{\sigma_{i-1}}} [Y_i \neq Y'_i]$ for good $w_{i-1}$ (notice that we assume nothing about $\sigma_{i-1}$).
A simple calculation shows that for $\kappa<1/2$, if $$\beta_i ( w_{i-1} ) \geq \left( \kappa \log |A| -\kappa\log \kappa- (1 - \kappa ) \log (1 - \kappa) \right) q_i (w_{i-1}),$$ then the probability mass (under $\cD(w_{i-1})$) on the most common label is at most $1-\kappa$. Observe that this probability is an upper bound on  $\Pr_{\cD_{w_{i-1}} \times \cD_{\sigma_{i-1}}} [Y_i = Y'_i]$.
For good $w_{i-1}$, we indeed have
\begin{gather*}
\beta_i (w_{i-1}) \geq \eps_3 q_i (w_{i-1}) \geq \left( \eps_4 \log |A| -\eps_4\log \eps_4- (1 - \eps_4) \log (1 - \eps_4) \right) q_i (w_{i-1}),
\end{gather*}
where the second inequality follows from choice of $\eps_4$.
Therefore $\Pr_{\cD_{w_{i-1}} \times \cD_{\sigma_{i-1}}} [Y_i \neq Y'_i] \geq \eps_4$.

We now have, for every good index $i$, 
\begin{align*}
\den(G_i) - \den(G_{i+1}) & \geq \sum_{\mbox{good \ensuremath{w_{i-1}}'s}} \sum_{\sigma_{i-1} \in \cN (w_{i-1})} \Big(\Pr[w_{i-1}] \Pr[\sigma_{i-1}] q_{i}(w_{i-1})  q_{i}(\sigma_{i-1})\Big) \eps_4  && \text{(Eq. \eqref{eq:density_loss})}\\
& \geq \eps_4 q_i (1-\eps_2) \sum_{\mbox{good \ensuremath{w_{i-1}}'s}} \Pr[w_{i-1}] q_{i}(w_{i-1}) && \hspace{-3cm}\text{(Definition of good prefix)} \\
& \geq \frac{\eps_1 \eps_4 }{10 \log |A|}q_i^2 && \hspace{-3cm} \text{(Corollary \pref{cor:good_prefix} + $\eps_2<0.2$)} 
\end{align*}

Finally, summing over all good $i$'s, we get a negative density for $S$, which is, of course, a contradiction.
\begin{align*}
1-\den(S) & \geq \den(G_1) - \den(G_{n+1}) && \text{(Observation \ref{obs:density-monotone})}\\
& = \sum_i   \den(G_i) - \den(G_{i+1}) && \text{(telescoping sum)}\\
& \geq \sum_{\mbox{good \ensuremath{i}'s}}   \den(G_i) - \den(G_{i+1}) \\
& \geq \sum_{\mbox{good \ensuremath{i}'s}}   \left(\frac{\eps_1 \eps_4 }{10 \log |A|}\right)q_i^2 \\
& \geq \left(\frac{\eps_1^3 \eps_4 }{250 \log^3 |A|}\right) \rho^2 / n = \Omega(\rho^2 /n). && \text{(Claim \ref{cla:goodvar})} \\
\end{align*} 
\end{proof}

% !TeX root = Dks_main.tex

\subsection{\TCSP violation}

Intuitively, if $\sum_i H(X_i | X_{ < i }, Y_{ < i })$ is large, then the subgraph approximately
corresponds to assignments to all subsets in ${\left[n\right] \choose \rho}$.
More specifically, in this section we show that most of the constraints appear approximately
as frequently as we expect. Since in any assignment, a constant fraction
of them must be violated, this implies (eventually) that a constant
fraction of the edges have a violated constraint.

First, we show that most of the variables appear approximately as frequently as we expect by showing that most of them are ``typical.''

\begin{definition}[Typical variables]
Prefix $w_{i-1}$ is {\em typical} if 
\begin{equation*}
\left(1-\eps_5\right)\cdot\rho/n < \Pr[ X_i = 1 | w_{i-1} ] < \left(1 + \eps_5\right)\cdot \rho / n,
\end{equation*}
where $\eps_5$ is some constant such that $\left(\frac{\log\e}{8}\right)\eps_5^{4}> 14 \eps_{1}$. \\
%Let $S_{i-1}$ be the set of typical prefixes for a fixed $i$. 
Similarly, we say that variable $x_i$ is {\em typical} if 
\begin{equation*}
\sum_{\mbox{typical \ensuremath{w_{i-1}}'s}} \Pr[w_{i-1}] \geq 1-\eps_5
\end{equation*}
\end{definition}

\begin{claim}\label{cla:many_typical_variables} 
If $\sum_i H(X_i | X_{ < i }, Y_{ < i }) \geq \left( 1 - \frac{6 \eps_1}{\log n} \right) \log k = \log k - \Theta(\rho)$, then all but at most $\eps_5 n$ variables are typical. 
\end{claim}
\begin{proof}
Assume by contradiction that there are $\eps_5 n$ atypical
variables. That is $\eps_5 n/2$ variables $x_{i}$ appear with
probability at least $\left(1+\eps_5\right)\cdot\rho/n$ (or
at most $\left(1-\eps_5\right)\cdot\rho/n$) for an $\left(\eps_5/2\right)$-fraction
of the prefixes $w_{i-1}$. Now, subject only to this constraint
and maintaining the correct expected number of variables in each vertex,
the entropy is maximized by spreading the $\left(\eps_5^{3}/4\right)$-loss
in frequency evenly across all other prefixes and variables. That is 
on the atypical prefixes, labels are assigned with probability $\left(1+\eps_5\right)\rho/n$,
and with probability $\left(1-\frac{\eps_5^{3}/4}{1-\eps_5^{2}/4}\right)\rho/n$
on the rest. Thus, 
\[
\sum_i H(X_i | X_{ < i }, Y_{ < i }) <\frac{\eps_5^{2}}{4}n\cdot H\left(\left(1+\eps_5\right)\rho/n\right)+\left(1-\frac{\eps_5^{2}}{4}\right)nH\left(\left(1-\frac{\eps_5^{3}/4}{1-\eps_5^{2}/4}\right)\rho/n\right)
\]
Recall from \pref{fact:entropy-taylor} the expansion of the entropy
function:
\[
H\left(\frac{1+\upsilon}{n}\right)=H\left(\frac{1}{n}\right)-\frac{\upsilon}{n}\log\frac{1}{n}-\left(\frac{\log\e}{2}\right)\frac{\upsilon^{2}}{n}+O\left(n^{-2}\right)+O\left(\frac{\upsilon^{3}}{n}\right)
\]
Therefore, 
\begin{eqnarray*}
\sum_i H(X_i | X_{ < i }, Y_{ < i }) & < & \frac{\eps_5^{2}}{4}n\left[{\color{blue} H\left(\frac{\rho}{n}\right)}-{\color{OliveGreen}\eps_5\frac{\rho}{n}\log\frac{\rho}{n}}-{\color{Maroon}\left(\frac{\log\e}{2}\right)\frac{\rho}{n}\cdot\eps_5^{2}}+O\left(\left(\frac{\rho}{n}\right)^{2}\right)+O\left(\frac{\rho}{n}\eps_5^{3}\right)\right]\\
 &  & +\left(1-\frac{\eps_5^{2}}{4}\right)n\left[{\color{blue}H\left(\frac{\rho}{n}\right)}+{\color{OliveGreen}\left(\frac{\eps_5^{3}/4}{1-\eps_5^{2}/4}\right)\frac{\rho}{n}\log\frac{\rho}{n}}+O\left(\left(\frac{\rho}{n}\right)^{2}\right)+O\left(\frac{\rho}{n}\eps_5^{6}\right)\right]\\
 & = & n\left[{\color{blue}H\left(\frac{\rho}{n}\right)}-{\color{Maroon}\left(\frac{\log\e}{8}\right)\frac{\rho}{n}\cdot\eps_5^{4}}+O\left(\left(\frac{\rho}{n}\right)^{2}\right)+O\left(\frac{\rho}{n}\eps_5^{5}\right)\right]
\end{eqnarray*}
Recall that $-2\log\frac{\rho}{n}<\log n$. Thus for $\left(\frac{\log\e}{8}\right)\eps_5^{4}> 14 \eps_{1}$,
we have that
\[
\left(\frac{\log\e}{8}\right)\frac{\rho}{n}\cdot\eps_5^{4}-O\left(\left(\frac{\rho}{n}\right)^{2}\right)-O\left(\frac{\rho}{n}\eps_5^{5}\right)>\frac{\rho}{n}\cdot 12 \eps_{1}>-\frac{\rho}{n}\log\frac{\rho}{n}\cdot 24 \eps_{1}/\log n>\left(12\eps_{1}/\log n\right)H\left(\frac{\rho}{n}\right)
\]
and therefore,
\[
\sum_i H(X_i | X_{ < i }, Y_{ < i }) <\left(1-12\eps_{1}/\log n\right)nH\left(\frac{\rho}{n}\right)<\left(1-6\eps_{1}/\log n\right)\log k,
\]
where the second inequality follows from Fact~\ref{fact:logk-expansion}. Thus we have reached a contradiction.
Notice that the $\left(\frac{\log\e}{8}\right)\frac{\rho}{n}\cdot\eps_5^{4}$
term of missing entropy is symmetric (but not the negligible higher
order terms); i.e. the same derivation can be used to show a contradiction
when many variables appear with probability less than $\left(1-\eps_5\right)\rho/n$.
\end{proof}

\begin{definition}
Let $\cI (u,v)$ be defined as the number of $(i,j)$ pairs such that
\begin{itemize}
\item In the original \TCSP instance $\psi$, there exists an edge (constraint) between typical variables $x_i$ and $x_j$. 
\item $X_i = 1$ for $u$ and $X_j = 1$ for $v$.
\item $u_{i-1}$ and $v_{j-1}$ are typical prefixes, where $u_{i-1}$ denotes the prefix represented by $u$ for $X_{<i}, Y_{<i}$, similarly for $v_{j-1}$.
\end{itemize}
\end{definition}

Intuitively, $\cI (u,v)$ is the number of ``tests'' of \TCSP-constraints between vertices $u,v$, when restricting to typical prefixes and variables.
We now use the properties of typical prefixes and constraints to show that $\cI (u,v)$ behaves ``nicely''.

\begin{claim} \label{cla:I(u,v)-moments}
$\E_{u,v}\left[\cI \left(u,v\right)\right]\geq\left(1-\eps_7\right)\rho^{2}/n$
and $\E_{u,v}\left[\cI^{2}\left(u,v\right)\right]\leq\left(1+ 2 \eps_7\right) d^4 \left(\E_{u,v}\left[\cI \left(u,v\right)\right]\right)^{2}$, where $\eps_7$ is some constant $\eps_7 \geq 6 \eps_5 + \Theta(\eps^2_5)$. 
\end{claim}
\begin{proof}
For any $i,j\in\left[n\right]$, we say that $i\in{\cal N}^{2CSP}\left(j\right)$
if there is a constraint on $\left(x_{i},x_{j}\right)$. 
For the proof of this claim, we also abuse notation and denote
$i\in v$ when $i$ is typical, $v_{i-1}$ is a typical prefix, and $X_i=1$ for $v$.
We also say that $i\in{\cal N}\left(u\right)$ if $i$ is a typical
variable, $i\in{\cal N}^{2CSP}\left(j\right)$, and $j\in u$ (for
some $j\in\left[n\right]$). 
(Do not confuse this notation with prefix neighborhood in the prefix graph.)
We can now lower bound the expectation of  $\cI \left(u,v\right)$ as:
\begin{equation*}
\E_{u,v}\left[\cI \left(u,v\right)\right]  \geq  \E_{u} \left[ \sum_{i\in{\cal N}\left(u\right)}\Pr_{v}\left[i\in v\right] \right]
\end{equation*}
Notice that this bound may not be tight since any $i \in v$ can potentially have $d$ neighbors in $u$. Thus our upper bound is:
\begin{equation*}
\E_{u,v}\left[\cI \left(u,v\right)\right]  \leq d \cdot \E_{u} \left[ \sum_{i\in{\cal N}\left(u\right)}\Pr_{v}\left[i\in v\right] \right]
\end{equation*}

By definition of typical variables, for each typical $i$, $i \in v$
with probability at least $\left(1-\eps_5\right)^{2}\rho/n$;
thus,
\begin{equation} \label{eq:lower}
\E_{u,v}\left[\cI \left(u,v\right)\right]\geq\E_{u} \left[ \sum_{i\in{\cal N}\left(u\right)}\left(1-\eps_5\right)^{2}\rho/n \right] =\left(1-\eps_5\right)^{2}\rho/n\cdot\E_{u} \left[ \left|{\cal N}\left(u\right)\right| \right]
\end{equation}
All but $\eps_5n$ variables are typical, so all but $2\eps_5n$
variables are typical and have at least one typical neighbor. We restrict
our attention to the set of such variables and fix one typical neighbor for each;
this neighbor appears in $u$ with probability at least $\left(1-\eps_5\right)^2\rho/n$. 
Therefore, 
\begin{gather} \label{eq:neighbor} 
\E_u \left[ | \cN (u) | \right] 
%&\geq \sum_{\text{$i$ and its neighbors are typical}} 1-\prod_{j\in{\cal N}^{2CSP}(i)} \left(1-\Pr[j \in u] \right) \nonumber \\
\geq ( 1- 2\eps_5 ) n \cdot \left(  (1 - \eps_5)^2 \rho / n \right)
\geq ( 1- 4\eps_5 ) \rho 
\end{gather}
Combining \pref{eq:lower} and \pref{eq:neighbor}, we get the desired bound:
\begin{equation}
\E_{u,v}\left[\cI \left(u,v\right)\right]\geq\left(\left(1-\eps_5\right)^{2}\rho/n\right)( 1- 4\eps_5 ) \rho \geq\left(1-\eps_7\right)\rho^{2}/n. \label{eq:E=00005BI=00005D}
\end{equation} %%%%%% 

Similarly, for the variance we have 
\begin{eqnarray*}
\E_{u,v}\left[\cI^{2}\left(u,v\right)\right] & \leq & d^2 \cdot\E_{u,v}\left(\sum_{i\in v\cap{\cal N}\left(u\right)}1\right)^{2}\\
 & = & d^2 \cdot \E_{u,v}\left[\sum_{i\neq j\in v\cap{\cal N}\left(u\right)}1+\sum_{i\in v\cap{\cal N}\left(u\right)}1\right]\\
 & \leq & d^2 \cdot \E_{u}\left[2\sum_{i<j\in{\cal N}\left(u\right)}\Pr_{v}\left[i\in v\right]\Pr_{v}\left[j\in v\mid i\in v\right]\right]+ d^2 \cdot\E_{u,v}\left[\cI \left(u,v\right)\right].
\end{eqnarray*}
Since for every prefix, each variable receives a typical assignment
with probability at most $\left(1+\eps_5\right)\cdot\rho/n$,
we have that
\begin{eqnarray}
\E_{u,v}\left[\cI^{2}\left(u,v\right)\right] & \leq & 2d^2 \cdot \E_{u}\left[\sum_{i<j\in{{\cal N}}\left(u\right)}\left(\left(1+\eps_5\right)\cdot\rho/n\right)^{2}\right]+d^2 \cdot \E_{u,v}\left[\cI\left(u,v\right)\right]\nonumber \\
 & \leq & \left(\left(1+\eps_5\right)\cdot\rho/n\right)^{2}\cdot 2d^2 \cdot \E_{u}{|{{\cal N}}\left(u\right)| \choose 2}+d^2\cdot\E_{u,v}\left[\cI\left(u,v\right)\right]\label{eq:E=00005BI^2=00005D}
\end{eqnarray}
We would like to bound $\E_{u}{{\cal N}\left(u\right) \choose 2}$.
\begin{eqnarray}
\E_{u}{{\cal N}\left(u\right) \choose 2} & = & \sum_{i<j}\sum_{k\in{\cal N}^{2CSP}\left(i\right)}\sum_{l\in{\cal N}^{2CSP}\left(j\right)}\Pr_{u}\left[k\in u\right]\Pr_{u}\left[l\in u\mid k\in u\right]\nonumber \\
 & = & \sum_{i<j}\sum_{\substack{k\in{\cal N}^{2CSP}\left(i\right)\\
l\in{\cal N}^{2CSP}\left(j\right)\\
\mbox{and }k<l
}
}\Pr_{u}\left[k\in u\right]\Pr_{u}\left[l\in u\mid k\in u\right]\label{eq:k<l}\\
 &  & +\sum_{i<j}\sum_{\substack{k\in{\cal N}^{2CSP}\left(i\right)\\
l\in{\cal N}^{2CSP}\left(j\right)\\
\mbox{and }k>l
}
}\Pr_{u}\left[l\in u\right]\Pr_{u}\left[k\in u\mid l\in u\right]\label{eq:k>l}\\
 &  & +\sum_{i<j}\sum_{k\in{\cal N}^{2CSP}\left(i\right)\cap{\cal N}^{2CSP}\left(j\right)}\Pr_{u}\left[k\in u\right]\label{eq:k-euqals-l}
\end{eqnarray}
For the first two summands, we can use the condition on the prefixes
to conclude that
\[
\eqref{eq:k<l}+\eqref{eq:k>l}\leq{n \choose 2}d^{2}\left(\left(1+\eps_5\right)\cdot\rho/n\right)^{2}
\]
Whereas to bound the third summand we first change the order of summation:
\begin{eqnarray*}
\eqref{eq:k-euqals-l} & = & \sum_{k}\Pr_{u}\left[k\in u\right]\cdot\left|\left\{ \left(i,j\right)\colon i\neq j\mbox{ and }k\in{\cal N}^{2CSP}\left(i\right)\cap{\cal N}^{2CSP}\left(j\right)\right\} \right|\\
 & \leq & \left(\left(1+\eps_5\right)\cdot\rho\right){d \choose 2}=O\left(\rho\right)
\end{eqnarray*}
Summing the last two inequalities, we have
\[
2 \cdot \E_{u}{\left| {{\cal N}}\left(u\right) \right| \choose 2}\leq d^{2}\left(\left(1+\eps_5\right)\cdot\rho\right)^{2}+O\left(\rho\right)\leq\left(1+\eps_5\right)^{3}d^{2}\rho^{2}
\]
Plugging back into \eqref{eq:E=00005BI^2=00005D}:
\begin{eqnarray*}
\E_{u,v}\left[\cI^{2}\left(u,v\right)\right] & \leq & \left(1+\eps_5\right)^{5}d^{4}\rho^{4}/n^{2}+d^2\cdot\E_{u,v}\left[\cI \left(u,v\right)\right]
\end{eqnarray*}
Using \eqref{eq:E=00005BI=00005D} and the fact that $\rho = \sqrt{n} \log \log n \gg\sqrt{n}$, this gives
\begin{align*}
\E_{u,v}\left[\cI^{2}\left(u,v\right)\right] &\leq \frac{d^4 (1+ \eps_5)^5}{1 - \eps_7} \left(\E_{u,v}\left[\cI \left(u,v\right)\right]\right)^{2} + d^2\cdot\E_{u,v}\left[\cI \left(u,v\right)\right] \\
& \leq \left(1+ 2 \eps_7\right) d^4 \left(\E_{u,v}\left[\cI \left(u,v\right)\right]\right)^{2}
\end{align*} 
\end{proof}

It will also be convenient to count the number of tests between a pair of variables.

\begin{definition}
For any pair of typical $(i,j) \in \psi$, let $\cI^\top (i,j)$ be defined as the number of $(u,v) \in (S \times S)$ pairs such that
\begin{itemize}
%\item In the original \TCSP instance $\psi$, there exists an edge (constraint) between $x_i$ and $x_j$. 
\item $X_i = 1$ for $u$ and $X_j = 1$ for $v$.
\item $u_{i-1}$ and $v_{j-1}$ are typical prefixes, where $u_{i-1}$ denotes the prefix represented by $u$ for $X_{<i}, Y_{<i}$, similarly for $v_{j-1}$.
\end{itemize}
\end{definition}

We now have two ways to count the total number of tests between typical prefixes to typical variables:
\begin{observation}\label{obs:sum(cI)=sum(cI^top)}
$\sum_{(u,v) \in (S \times S)} \cI(u,v) = \sum_{(i,j) \in \psi} \cI^\top (i,j)$.
\end{observation}

Furthermore, since $i$ and $j$ are typical, the number of tests between also behaves ``nicely'':

\begin{observation}
\label{obs:cI^top = rho^2/n^2}
For every typical $(i,j) \in \psi$, we have $\cI^\top(i,j)  \in |S|^2 \rho^2/n^2 \Big[\left(1-\eps_5\right)^4, \left(1 + \eps_5\right)^2 \Big]$.
\end{observation}
\begin{proof}
\begin{align*}
\cI^\top(i,j) &= \sum_{\mbox{typical \ensuremath{u_{i-1}}'s}} |S| \cdot  \Pr[u_{i-1}]\Pr\left[X_{i}=1\mid u_{i-1}\right]
\sum_{\mbox{typical \ensuremath{v_{j-1}}'s}} |S|\cdot  \Pr\left[v_{j-1}\right]\Pr\left[X_{j}=1\mid v_{j-1}\right]\\
& \in |S|^2 \rho^2/n^2 \Big[\left(1-\eps_5\right)^4, \left(1 + \eps_5\right)^2 \Big]
\end{align*}
\end{proof}
\vspace{0.5cm}

\ignore{
\begin{align*}
\cI^\top_{\mbox{mode}} (i,j)\geq  \cI^\top(i,j) / |A|^2
\end{align*}

\vspace{0.5cm}

\begin{align*}
\sum_{\mbox{typcial, unsatisfied \ensuremath{(i,j)}'s}} \cI^\top_{\mbox{mode}} (i,j)  
& \geq \sum_{\mbox{typical, unsatisfied \ensuremath{(i,j)}'s}} \cI^\top (i,j)  / |A|^2 \\
& \geq  \frac{\left(1-\eps_5\right)^4}{ \left(1 + \eps_5\right)^2}   
\cdot \frac{\big|\big\{\mbox{typical, unsatisfied \ensuremath{(i,j)}'s}\big\}\big|}
{\big|\big\{\mbox{typical \ensuremath{(i,j)\in \psi}}\big\} \big|}
\cdot \sum_{(i,j) \in \psi} \cI^\top (i,j) / |A|^2 \\
\end{align*}

\newpage
}

Armed with these Claims \ref{cla:many_typical_variables} and \ref{cla:I(u,v)-moments} and Observations \ref{obs:sum(cI)=sum(cI^top)} and \ref{obs:cI^top = rho^2/n^2}, we are now ready to prove the main lemma of this section. Recall that the soundness of the 2CSP we started with is 
$1-\eta$ for a small constant $\eta$. 

%%%%%%%%%%

\begin{lemma} \label{lem:2cspvio}
If $\sum_i H(X_i | X_{ < i }, Y_{ < i }) \geq \left( 1 - \frac{6\eps_1}{\log n} \right) \log k$, then $\delta(S) < 1 - \delta$, where $\delta < \frac{\eps^2_6}{d^4(1+2\eps_7)}$ and $\eps_6=\left(\eta/2-\eps_5\right)\left(1/|A|^{2}\right)\frac{\left(1-\eps_5\right)^{4}}{\left(1+\eps_5\right)^{2}}$. 
\end{lemma}

\begin{proof}
Let the {\em mode assignment} be the assignment $\cA\colon\left[n\right]\rightarrow\Sigma$
which assigns to each variable $x_{i}$ its most common typical assignment 
(i.e. assignment after a typcial prefix),
breaking ties arbitrarily. In particular, at least $1/|A|$ of
the typical assignments for $x_{i}$ are equal to $\cA\left(i\right)$.
Of course, this assignment cannot satisfy more than a $\left(1-\eta\right)$-fraction
of the constraints in the original \TCSP; after removing the $\eps_5n$
atypical variables, $\left(\eta/2-\eps_5\right)d n$ constraints out of the $d n/2$ constraints
must still be unsatisfied. 

Recall that the number of tests for each constraint over typcial variables, 
$\cI^\top(i,j)$, is approximately the same for every pair of $(i,j)$ --- up to a $\frac{\left(1-\eps_5\right)^{4}}{\left(1+\eps_5\right)^{2}}$-multiplicative factor  (Observation \ref{obs:cI^top = rho^2/n^2}).
Therefore, the total fraction of tests over unsatisfied constraints, out of all tests, 
is approximately proportional to the fraction of unsatisfied constraints:

\begin{align*}
\sum_{\mbox{typcial, unsatisfied \ensuremath{(i,j)}'s}} \cI^\top (i,j)  
& \geq  \frac{\left(1-\eps_5\right)^4}{ \left(1 + \eps_5\right)^2}   
\cdot \frac{\big|\big\{\mbox{typical, unsatisfied \ensuremath{(i,j)}'s}\big\}\big|}
{\big|\big\{\mbox{typical \ensuremath{(i,j)\in \psi}}\big\} \big|}
\cdot \sum_{(i,j) \in \psi} \cI^\top (i,j)\\
&\geq \frac{\left(1-\eps_5\right)^4}{ \left(1 + \eps_5\right)^2}   
\cdot \frac{\left(\eta/2 - \eps_5\right)dn}{dn/2}
\cdot \sum_{(i,j) \in \psi} \cI^\top (i,j)\\
& = \frac{\left(1-\eps_5\right)^4}{ \left(1 + \eps_5\right)^2}   
\cdot \left(\eta - 2\eps_5\right) 
\cdot \sum_{(u,v) \in (S \times S)} \cI(u,v) && \hspace{-1.7cm}\text{(Observation \ref{obs:sum(cI)=sum(cI^top)})}
\end{align*}

For each such pair $(i,j)$, on at least a $1/|A|^2$-fraction of the tests
both variables receive the mode assignment, so the constraint is violated\footnote{We remark that a more careful analysis of the expected number of violations would allow one to save an $|A|^2$-factor in the value of $\varepsilon_6$. Since it does not qualitatively affect the result, we opt for the simpler analysis.}.
Thus the total number of violations is at least $\eps_6\sum_{(u,v) \in (S \times S)}\cI\left(u,v\right)$
(where $\eps_6=\left(\eta/2-\eps_5\right)\left(1/|A|^{2}\right)\frac{\left(1-\eps_5\right)^{4}}{\left(1+\eps_5\right)^{2}}$).

Finally, we show that so many violations cannot concentrate
on less than a $\delta$-fraction of the pairs $u,v\in S$;
otherwise:
\begin{align*}
\sum_{\left(u,v\right)\in\left(S\times S\right)\setminus E}\cI^{2}\left(u,v\right) & \geq \frac{1}{\delta\left|S\right|^{2}}\left(\sum_{\left(u,v\right)\in\left(S\times S\right)\setminus E}\cI\left(u,v\right)\right)^{2}   && \text{(Cauchy-Schwartz)}\\
 & \geq \frac{1}{\delta\left|S\right|^{2}}\left(\eps_6\sum_{\left(u,v\right)\in \left(S \times S\right)}\cI\left(u,v\right)\right)^{2}\\
 & = \frac{|S|^{2}\eps_6^{2}}{\delta}\left(\E_{u,v}\left[\cI\left(u,v\right)\right]\right)^2;
\end{align*}
%Where first inequality follows from Cauchy-Schwarz. 
yet by Claim \ref{cla:I(u,v)-moments}, 
\[
\sum_{\left(u,v\right)\in\left(S\times S\right)\setminus E}\cI^{2}\left(u,v\right)\le \sum_{\left(u,v\right)\in S\times S}\cI^{2}\left(u,v\right)\leq\left(1+2\eps_7\right) d^4 |S|^{2}\left(\E_{u,v}\left[\cI\left(u,v\right)\right]\right)^{2}.
\]
Thus we have a contradiction since $d^4 (1+2\eps_7)<\eps_6^{2}/\delta$ by our setting of $\delta$. 
Therefore we have \TCSP-violations in more than a $\delta$-fraction of the pairs $u,v\in S$.
\end{proof}

With \pref{lem:consistency} and \pref{lem:2cspvio}, we can now complete the proof of \pref{thm:soundness}. 

\thmSoundness*
\ignore{
\begin{reptheorem}{thm:soundness}
If $\OPT(\psi) < 1 - \eta$, then $\forall S \subset V$ of size $k$, $\delta(S) < 1 - \delta$ where $\delta$ is some small enough constant that satisfies $\delta < \frac{\eps^2_6}{d^4(1+2\eps_7)}$ and $\delta < \left( \frac{\eps_2}{|A|^{2/\eps_2}} \right)^4$ 
\end{reptheorem}
}

\begin{proof}
Recall that $\sum_i \alpha_i + \beta_i = \log k' \geq (1-\frac{\eps_1}{\log n})\log k$ by Fact~\ref{fact:logkprime-expansionNEW}.
If $\sum_i \beta_i > (\frac{5\eps_1 }{\log n}) \log k$, then by~\pref{lem:consistency}, $\delta(S) < 1- \delta$. Otherwise,
if  $\sum_i \alpha_i > (1-\frac{6\eps_1}{\log n}) \log k$, by~\pref{lem:2cspvio}, $\delta(S) < 1 - \delta$.
\end{proof}

\bibliographystyle{alpha}
\bibliography{refs}

\newcommand{\etalchar}[1]{$^{#1}$}
\begin{thebibliography}{BPR{\etalchar{+}}15b}

\bibitem[AAK{\etalchar{+}}07]{AlonAKMRX07}
Noga Alon, Alexandr Andoni, Tali Kaufman, Kevin Matulef, Ronitt Rubinfeld, and
  Ning Xie.
\newblock Testing k-wise and almost k-wise independence.
\newblock In {\em STOC}, pages 496--505. ACM, 2007.

\bibitem[AAM{\etalchar{+}}11]{AAMMW11}
Noga Alon, Sanjeev Arora, Rajsekar Manokaran, Dana Moshkovitz, and Omri
  Weinstein.
\newblock Inapproximability of densest $\kappa$-subgraph from average case
  hardness.
\newblock {\em Unpublished manuscript}, 2011.

\bibitem[AGM13]{arora2013new}
Sanjeev Arora, Rong Ge, and Ankur Moitra.
\newblock New algorithms for learning incoherent and overcomplete dictionaries.
\newblock {\em arXiv preprint arXiv:1308.6273}, 2013.

\bibitem[AGSS12]{arora2012finding}
Sanjeev Arora, Rong Ge, Sushant Sachdeva, and Grant Schoenebeck.
\newblock Finding overlapping communities in social networks: toward a rigorous
  approach.
\newblock In {\em Proceedings of the 13th ACM Conference on Electronic
  Commerce}, pages 37--54. ACM, 2012.

\bibitem[AIM14]{AIM14}
Scott Aaronson, Russell Impagliazzo, and Dana Moshkovitz.
\newblock Am with multiple merlins.
\newblock In {\em Computational Complexity (CCC), 2014 IEEE 29th Conference
  on}, pages 44--55. IEEE, 2014.

\bibitem[AKS98]{AlonKS98}
Noga Alon, Michael Krivelevich, and Benny Sudakov.
\newblock Finding a large hidden clique in a random graph.
\newblock In {\em SODA}, pages 594--598, 1998.

\bibitem[ALM{\etalchar{+}}98]{ALMSS98-PCP}
Sanjeev Arora, Carsten Lund, Rajeev Motwani, Madhu Sudan, and Mario Szegedy.
\newblock Proof verification and the hardness of approximation problems.
\newblock {\em J. {ACM}}, 45(3):501--555, 1998.

\bibitem[AS98]{AS98-PCP}
Sanjeev Arora and Shmuel Safra.
\newblock Probabilistic checking of proofs: {A} new characterization of {NP}.
\newblock {\em J. {ACM}}, 45(1):70--122, 1998.

\bibitem[Bar15a]{Barman15-DkS_vs_DkBS}
Siddharth Barman.
\newblock personal communication, 2015.

\bibitem[Bar15b]{Barman14}
Siddharth Barman.
\newblock Approximating carath{\'{e}}odory's theorem and nash equilibria.
\newblock In {\em STOC}, 2015.

\bibitem[BBB{\etalchar{+}}13]{balcan2013finding}
Maria-Florina Balcan, Christian Borgs, Mark Braverman, Jennifer Chayes, and
  Shang-Hua Teng.
\newblock Finding endogenously formed communities.
\newblock In {\em Proceedings of the Twenty-Fourth Annual ACM-SIAM Symposium on
  Discrete Algorithms}, pages 767--783. SIAM, 2013.

\bibitem[BCV{\etalchar{+}}12]{Bhaskara:2012:PIG:2095116.2095150}
Aditya Bhaskara, Moses Charikar, Aravindan Vijayaraghavan, Venkatesan
  Guruswami, and Yuan Zhou.
\newblock Polynomial integrality gaps for strong sdp relaxations of densest
  k-subgraph.
\newblock In {\em Proceedings of the Twenty-third Annual ACM-SIAM Symposium on
  Discrete Algorithms}, SODA '12, pages 388--405. SIAM, 2012.

\bibitem[BKW15]{BKW15}
Mark Braverman, Young~Kun Ko, and Omri Weinstein.
\newblock Approximating the best nash equilibrium in $n^{o(\log n) }$-time
  breaks the exponential time hypothesis.
\newblock In {\em ACM-SIAM Symposium on Discrete Algorithms (SODA)}, 2015.

\bibitem[BL13]{balcan2013modeling}
Maria~Florina Balcan and Yingyu Liang.
\newblock Modeling and detecting community hierarchies.
\newblock In {\em Similarity-Based Pattern Recognition}, pages 160--175.
  Springer, 2013.

\bibitem[BPR15a]{BPR15-PCP_for_PPAD}
Yakov Babichenko, Christos Papadimitriou, and Aviad Rubinstein.
\newblock {Can Almost Everybody be Almost Happy? PCP for PPAD and the
  Inapproximability of Nash}.
\newblock In submission, 2015.

\bibitem[BPR{\etalchar{+}}15b]{BPRSS15-seeding}
Ashwinkumar Badanidiyuru, Christos Papadimitriou, Aviad Rubinstein, Lior
  Seeman, and Yaron Singer.
\newblock Submodular adaptive seeding.
\newblock In submission, 2015.

\bibitem[BR13]{berthet2013complexity}
Quentin Berthet and Philippe Rigollet.
\newblock Complexity theoretic lower bounds for sparse principal component
  detection.
\newblock In {\em Conference on Learning Theory}, pages 1046--1066, 2013.

\bibitem[CLLR15]{CLLR15-amphibious}
Wei Chen, Fu~Li, Tian Lin, and Aviad Rubinstein.
\newblock Combining traditional marketing and viral marketing with amphibious
  influence maximization.
\newblock In submission, 2015.

\bibitem[CT12]{CT91}
Thomas~M Cover and Joy~A Thomas.
\newblock {\em Elements of information theory}.
\newblock John Wiley \&amp; Sons, 2012.

\bibitem[DGGP10]{DekelGP10}
Yael Dekel, Ori Gurel-Gurevich, and Yuval Peres.
\newblock Finding hidden cliques in linear time with high probability.
\newblock {\em CoRR}, abs/1010.2997, 2010.

\bibitem[Din07]{dinur2007pcp}
Irit Dinur.
\newblock The pcp theorem by gap amplification.
\newblock {\em Journal of the ACM (JACM)}, 54(3):12, 2007.

\bibitem[DM15]{DM15-SoS-hidden_submatrix}
Yash Deshpande and Andrea Montanari.
\newblock Improved sum-of-squares lower bounds for hidden clique and hidden
  submatrix problems.
\newblock {\em CoRR}, abs/1502.06590, 2015.

\bibitem[Fei02]{Feige02}
Uriel Feige.
\newblock Relations between average case complexity and approximation
  complexity.
\newblock In {\em STOC}, pages 534--543. ACM Press, 2002.

\bibitem[FGL{\etalchar{+}}96]{fglss96}
Uriel Feige, Shafi Goldwasser, Laszlo Lov{\'a}sz, Shmuel Safra, and Mario
  Szegedy.
\newblock Interactive proofs and the hardness of approximating cliques.
\newblock {\em Journal of the ACM (JACM)}, 43(2):268--292, 1996.

\bibitem[FGR{\etalchar{+}}13]{FGRVX13-statistical}
Vitaly Feldman, Elena Grigorescu, Lev Reyzin, Santosh Vempala, and Ying Xiao.
\newblock Statistical algorithms and a lower bound for detecting planted
  cliques.
\newblock In {\em Symposium on Theory of Computing Conference, STOC'13, Palo
  Alto, CA, USA, June 1-4, 2013}, pages 655--664, 2013.

\bibitem[FK00]{FeigeK00}
Uriel Feige and Robert Krauthgamer.
\newblock Finding and certifying a large hidden clique in a semirandom graph.
\newblock {\em Random Struct. Algorithms}, 16(2):195--208, 2000.

\bibitem[FKP01]{FPK01}
Uriel Feige, Guy Kortsarz, and David Peleg.
\newblock The dense \emph{k}-subgraph problem.
\newblock {\em Algorithmica}, 29(3):410--421, 2001.

\bibitem[FS97]{fs97}
Uriel Feige and Michael Seltser.
\newblock {\em On the densest k-subgraph problem}.
\newblock Citeseer, 1997.

\bibitem[H{\aa}s99]{Hastad96-clique_n^1-eps}
Johan H{\aa}stad.
\newblock Clique is hard to approximate within n\({}^{\mbox{1-epsilon}}\).
\newblock {\em Acta Mathematica}, 182(1):105--142, 1999.

\bibitem[HK11]{HazanK11}
Elad Hazan and Robert Krauthgamer.
\newblock How hard is it to approximate the best nash equilibrium?
\newblock {\em SIAM J. Comput.}, 40(1):79--91, 2011.

\bibitem[IP01]{ETH01}
Russell Impagliazzo and Ramamohan Paturi.
\newblock On the complexity of k-sat.
\newblock {\em J. Comput. Syst. Sci.}, 62(2):367--375, 2001.

\bibitem[Jer92]{Jerrum92}
Mark Jerrum.
\newblock Large cliques elude the metropolis process.
\newblock {\em Random Struct. Algorithms}, 3(4):347--360, 1992.

\bibitem[Kar72]{Karp72}
Richard~M. Karp.
\newblock Reducibility among combinatorial problems.
\newblock In {\em Proceedings of a symposium on the Complexity of Computer
  Computations, held March 20-22, 1972, at the {IBM} Thomas J. Watson Research
  Center, Yorktown Heights, New York.}, pages 85--103, 1972.

\bibitem[Kho01]{Khot01-clique_coloring}
Subhash Khot.
\newblock Improved inaproximability results for maxclique, chromatic number and
  approximate graph coloring.
\newblock In {\em 42nd Annual Symposium on Foundations of Computer Science,
  {FOCS} 2001, 14-17 October 2001, Las Vegas, Nevada, {USA}}, pages 600--609,
  2001.

\bibitem[Kho06]{Kho06}
Subhash Khot.
\newblock Ruling out ptas for graph min-bisection, dense k-subgraph, and
  bipartite clique.
\newblock {\em SIAM Journal on Computing}, 36(4):1025--1071, 2006.

\bibitem[Kuc95]{Kucera95}
Ludek Kucera.
\newblock Expected complexity of graph partitioning problems.
\newblock {\em Discrete Applied Mathematics}, 57(2-3):193--212, 1995.

\bibitem[MPW15]{MPW15-SOS_for_planted_clique}
Raghu Meka, Aaron Potechin, and Avi Wigderson.
\newblock Sum-of-squares lower bounds for planted clique.
\newblock In {\em STOC}, 2015.

\bibitem[RS10]{RS10}
Prasad Raghavendra and David Steurer.
\newblock Graph expansion and the unique games conjecture.
\newblock In {\em Proceedings of the forty-second ACM symposium on Theory of
  computing}, pages 755--764. ACM, 2010.

\bibitem[Zuc07]{Zuckerman07-clique_NP}
David Zuckerman.
\newblock Linear degree extractors and the inapproximability of max clique and
  chromatic number.
\newblock {\em Theory of Computing}, 3(1):103--128, 2007.

\end{thebibliography}

\appendix
% !TeX root = Dks_main.tex

\section{PCP theorem}

\thmTCSP*
\ignore{
\begin{theorem}
Given a \TSAT instance $\varphi$ of size $n$, there is a polynomial time reduction that produces a \TCSP 
instance $\psi$, with size $|\psi| = n \cdot \polylog n$ variables and constraints, and constant alphabet size,
such that
\begin{itemize}
\item (Completeness) If $\OPT( \varphi ) = 1$ then $\OPT( \psi ) = 1$.
\item (Soundness) If $\OPT( \varphi ) < 1 $ then $ \OPT(\psi) < 1 - \eta$, for some constant $\eta = \OPT(1)$
\item (Balance) Every vertex in $\psi$ has either degree 3 or degree $d$ for some constant $d$.
\end{itemize}
\end{theorem}}
\begin{proof}
We start with the following version of PCP of near linear size.

\begin{theorem}[\cite{dinur2007pcp}, version as in \cite{AIM14}] \label{thm_Dinur}
Given a \TSAT instance $\varphi$ of size $n$, there is a polynomial time reduction that produces a \TSAT 
instance $\xi$, with size $|\xi| = n \cdot \polylog n$ variables and constraints such that
\begin{itemize}
\item (Completeness) If $\OPT( \varphi ) = 1$ then $\OPT( \xi ) = 1$.
\item (Soundness) If $\OPT( \varphi ) < 1 $ then $ \OPT(\xi) < 1 - \eps$, for some constant $0 < \eps < 1/8$
\item (Balance) Every clause of $\psi$ involves exactly 3 variables, and every variable of $\psi$ appears in exactly $d$ clauses, for some constant $d$.
\end{itemize}
\end{theorem}

We use the following definition to reduce $\xi$ given by Theorem~\ref{thm_Dinur} to a \TCSP instance $\psi$. 

\begin{definition}[\cite{AIM14}, Clause/Variable game] \label{def_cv}
Given a \TSAT instance $\xi$ with $n$ variables $x_1, \ldots x_n$ and $m$ clauses $C_1, \ldots, C_m$, the clause/variable game $G_\xi$ is defined as follows:
Referee chooses an index $i \in m$ uniformly at random, then chooses $j \in [n]$ uniformly at random conditioned on $x_j$ or $\overline{x_j}$ appearing in $C_i$ as a literal. He sends $i$ to Alice and $j$ to Bob. Referee accepts if and only if
\begin{itemize}
\item Alice sends back a satisfying assignment to the variables in $C_i$.
\item Bob sends back a value for $x_j$ that agrees with the value sent by Alice.
\end{itemize}
\end{definition}

In particular, we can think of following explicit reduction.
\begin{enumerate}
\item $X = [m]$ represents clauses; $Y = [n]$ represents variables; $A = \{ 0, 1 \}^3 $ represents assignment to all 3 variables in a clause; $B = \{ 0, 1 \}$ represents assignment to a singleton variable.
\item $(i,j) \in E$ if $x_j$ or $\overline{x_j}$ appears in $i$th clause ($C_i$).
\item $V_{(i,j)}$ checks for the following :
\begin{itemize}
\item Assignment on $i \in [m]$ indeed satisfies the clause $C_i$ and, 
\item Assignment on $i \in [m]$ agrees with the assignment on $j \in [n]$. 
\end{itemize}
\end{enumerate}

The size blowup is indeed only constant, since we have linear number of vertices, and constant alphabet size.
Also any vertex in $X$ has degree 3, and any vertex in $Y$ has degree $d$ since we we started with Dinur's PCP.
Completeness follows by assigning satisfying assignment for \TSAT to this \TCSP. Soundness follows from the following claim: 
\begin{claim}[\cite{AIM14}]
$\OPT(\xi) \leq 1 - \eps$, then $\OPT(\psi) \leq 1 - \eps/3$
\end{claim}
\begin{proof}
Consider fixing an assignment $x$ on $Y$'s. By our assumption on $\xi$, this violates the clause $C_i$ with probability at least $\eps$ over $i$. And if $x$ violates $C_i$, regardless of assignments on $X$, at least one out of 3 edges of $i \in X$ is not satisfied. Therefore, at least $\eps/3$-fraction of the edges are violated, thus $\OPT(\psi) \leq 1 - \eps / 3$.
\end{proof}

Now we add trivial constraints (i.e. always satisfying edges) between vertices in $X$ to make the overall graph of $\psi$ $d$-regular. (we lose bipartite property, which is not necessary in our reduction) Take a regular graph on $X$ with degree $d - 3$. Add the edges with constraints on them as trivial constraints to our \TCSP instance $\psi$ generated via the reduction. Now the graph is indeed $d$-regular, completeness is preserved since we only added trivial constraints. For soundness, we know that there are now total $3|X| + \frac{d-3}{2} |X| $ edges. Among them $\frac{d-3}{2} |X|$ are always satisfied. Out of $3 |X|$ edges, at most $1 - \eps/3$ fraction of them are satisfied, i.e. $(3 - \eps) |X|$ edges. So the fraction of satisfied edges is at most :
\begin{equation*}
\OPT(\psi) \leq \frac{ (3 - \eps) |X| + \frac{d-3}{2} |X| }{ 3|X| + \frac{d-3}{2} |X|  } = \frac{d + 3 - 2 \eps}{d+3} \leq 1 - \frac{\eps}{d} = 1 - \eta 
\end{equation*}
\end{proof}

\section{Useful approximations}

We recall some elementary approximations to logarithms and entropies
that will be useful in the analysis.  

\begin{fact} {\bf (Fact~\ref{fact:logk-expansion})}
If $k = {n \choose \rho}$ then,
\begin{equation*}
\log k=nH\left(\frac{\rho}{n}\right)\pm O\left(\log n\right)=\left(\frac{1}{2}-o\left(1\right)\right)\rho\log n
\end{equation*}
\end{fact}
\begin{proof}
By Stirling's approximation, we have 
\begin{eqnarray*}
\log n! & = & n\log n-\left(\log\e\right)n+O\left(\log n\right)
\end{eqnarray*}
Therefore the total entropy is given by
\begin{eqnarray*}
\log k & = & \log{n \choose \rho}\\
 & = & \log n!-\log\rho!-\log\left(n-\rho\right)!\\
 & = & n\log n-\rho\log\rho-\left(n-\rho\right)\log\left(n-\rho\right)\pm O\left(\log n\right)\\
 & = & nH\left(\frac{\rho}{n}\right)\pm O\left(\log n\right)\mbox{,}
\end{eqnarray*}

For small $\epsilon$, we have 
\[
\log\left(1+\epsilon\right)=\left(\log\e\right)\left(\epsilon-\frac{\epsilon^{2}}{2}+O\left(\epsilon^{3}\right)\right)\mbox{;}
\]
and in particular, 
\[
\log\frac{n-\rho}{n}=O\left(-\frac{\rho}{n}\right)
\]
Therefore, 
\begin{align*}
\log k & =  \rho\cdot\log\frac{n}{\rho}+\left(n-\rho\right)\cdot\log\frac{n}{n-\rho}+O\left(\log n\right)\\
 & =  \rho\cdot\left(\frac{1}{2}-o\left(1\right)\right)\log n+\left(n-\rho\right)\cdot O\left(\frac{\rho}{n}\right)+O\left(\log n\right)\\
 & =  \left(\frac{1}{2}-o\left(1\right)\right)\rho\log n
\end{align*}

\end{proof}

More useful to us will be the following bounds on $\log k'$:

\begin{fact} {\bf (Fact~\ref{fact:logkprime-expansionNEW})}
%\label{fact:logkprime-expansion}
Let $\epsilon_1 \ge 5\epsilon_0$, and $k,k',V,n,\rho$ as specified in the construction. Then,
\begin{gather*}
\log k'  \geq \max \left\{\log k, nH\left(\frac{\rho}{n}\right) \right\} -\epsilon_1 \log k/\log n.
\end{gather*}

In particular, this means that most indices $i$ %(except the last few), we expect to
should contribute roughly $H\left(\frac{\rho}{n}\right)$ entropy to the choice of $v$.
\end{fact}

\begin{proof}

Observing that since $k =  {n \choose \rho}$, we have
\begin{gather}
\label{eq:logV}
\log |V| = \log {n\choose{\rho}} + \rho \log |A| = (1+o(1)) \log k.
\end{gather}
We also have that 
\begin{gather}
\label{eq:loglogV}
\log\log |V| = \log (1+o(1)) + \log \log k > \log \rho > \frac{1}{2}\log n;
\end{gather}
where the first inequality follows from Fact \ref{fact:logk-expansion}, and the second from the definition of $\rho$.

Finally, we have
\begin{align*}
\log k' &= \log k - \epsilon_0 \log |V| / \log \log |V|\\
& \geq \log k - \epsilon_0  (1+o(1)) \log k/ \frac{1}{2}\log n  && \text{(Using \eqref{eq:logV} and \eqref{eq:loglogV})}\\
& \geq \log k - \frac{1}{2}\epsilon_1 \log k / \log n && \text{(Using $\epsilon_1 \ge 5\epsilon_0$)} 
\end{align*}
Using Fact \ref{fact:logk-expansion} completes the proof.
\end{proof}

We will also need the following bound which relates the entropies of a very biased coin and a slightly less biased one:
%, for small deviation on binary entropy function from $1/n$. 
\begin{fact} {\bf (Fact~\ref{fact:entropy-taylor})}
%\label{fact:entropy-taylor}Let $1/n\ll\left|\upsilon\right|\ll1$, 
\[
H\left(\frac{1+\upsilon}{n}\right)=H\left(\frac{1}{n}\right)-\frac{\upsilon}{n}\log\frac{1}{n}-\left(\log\e\right)\frac{\upsilon^{2}}{2n}+O\left(n^{-2}\right)+O\left(\frac{\upsilon^{3}}{n}\right)
\]
\end{fact}
\begin{proof}
By definition, 
\[
H\left(\frac{1+\upsilon}{n}\right)=-\left(\frac{1+\upsilon}{n}\right)\log\left(\frac{1+\upsilon}{n}\right)-\left(1-\frac{1+\upsilon}{n}\right)\log\left(1-\frac{1+\upsilon}{n}\right)
\]
In order to relate this quantity to $H\left(\frac{1}{n}\right)$,
we rewrite as: 
\begin{eqnarray*}
H\left(\frac{1+\upsilon}{n}\right) & = & -\frac{1}{n}\log\frac{1}{n}-\frac{\upsilon}{n}\log\frac{1}{n}-\left(\frac{1+\upsilon}{n}\right)\cdot\underbrace{\log\left(1+\upsilon\right)}_{\left(\log\e\right)\left(\upsilon-\upsilon^{2}/2+O\left(\upsilon^{3}\right)\right)}\\
 &  & -\left(1-\frac{1}{n}\right)\log\left(1-\frac{1}{n}\right)+\upsilon\underbrace{\frac{1}{n}\log\left(1-\frac{1}{n}\right)}_{O\left(n^{-2}\right)}-\left(1-\left(\frac{1+\upsilon}{n}\right)\right)\cdot\underbrace{\log\left(\frac{1-\left(\frac{1+\upsilon}{n}\right)}{1-\frac{1}{n}}\right)}_{\left(\log\e\right)\left(-\left(\upsilon/n\right)-O\left(\upsilon/n^{2}\right)\right)}\\
 & = & H\left(\frac{1}{n}\right)-\frac{\upsilon}{n}\log\frac{1}{n}-\left(\log\e\right)\frac{\upsilon^{2}}{2n}+O\left(n^{-2}\right)+O\left(\frac{\upsilon^{3}}{n}\right)
\end{eqnarray*}
\end{proof}

% !TeX root = Dks_main.tex

\section{Small constants in the proof of Theorem \ref{thm:soundness}}
To help verify the correctness of the proof, we concentrate all the definitions of the small $\epsilon$'s used in the following list:
\begin{itemize}
\item $\epsilon_0 \leq  \epsilon_1 / 5$
\item $\eps_1 \geq 4 \eps_2 \log |A| + 8 \eps_3$
\item $\epsilon_2$: $\eps_2<0.2$, $\delta = \left( \frac{\eps_2}{|A|^{2/\eps_2}} \right)^4$
\item $\eps_3 \geq  \eps_4 \log |A| -\eps_4\log\eps_4- (1-\eps_4)\log(1-\eps_4)$
\item $\eps_4 = \omega(n/ \rho^2)$
\item $\eps_5$: $\left(\frac{\log\e}{8}\right)\eps_5^{4}> 14 \eps_{1}$
\item $\eps_6$: $ \eps_6 =\left(\eta/2-\eps_5\right)\left(1/|A|^{2}\right)\frac{\left(1-\eps_5\right)^{4}}{\left(1+\eps_5\right)^{2}}$ and 
$d^4 (1+2\eps_7)<\eps_6^{2}/\delta$
\item $\eps_7$: $\eps_7 \geq 6 \eps_5 + \Theta(\eps^2_5)$
\end{itemize}

\end{document}